\documentclass[journal]{IEEEtran}

\usepackage{amsmath}
\usepackage{amsfonts}
\usepackage{amssymb}
\usepackage{dsfont}
\usepackage{bm}
\usepackage{amsthm}
\usepackage{newlfont}
\usepackage{float}
\usepackage{hyperref}
\usepackage{enumerate}
\usepackage{chngcntr}
\usepackage{mathtools}
\usepackage{breqn}
\usepackage{stmaryrd}
\usepackage{cite}
\usepackage{flushend}

\usepackage{algorithm}
\usepackage{algorithmic}
\hypersetup{
    colorlinks=true, 
    linkcolor= blue, 
    citecolor=blue 
}
\allowdisplaybreaks

\begin{document}

\newtheorem{thm}{Theorem} 
\newtheorem{lem}{Lemma}
\newtheorem{prop}{Proposition}
\newtheorem{cor}{Corollary}
\newtheorem{defn}{Definition}
\newtheorem{rem}{Remark}
\newtheorem{ex}{Example}
\newenvironment{example}[1][Example]{\begin{trivlist}
\item[\hskip \labelsep {\bfseries #1}]}{\end{trivlist}}
\renewcommand{\qedsymbol}{ \begin{tiny}$\blacksquare$ \end{tiny} }
\renewcommand{\leq}{\leqslant}
\renewcommand{\geq}{\geqslant}

\title {Empirical and Strong Coordination\\via Soft Covering with Polar Codes}

%
\author{R\'{e}mi A. Chou,~\IEEEmembership{Member,~IEEE}, Matthieu R. Bloch,~\IEEEmembership{Senior Member,~IEEE},\\ and J\"{o}rg Kliewer,~\IEEEmembership{Senior Member,~IEEE}
\thanks{
R. A. Chou is with the Department of Electrical Engineering and Computer Science, Wichita State University, Wichita, KS. M. R. Bloch is with the School~of~Electrical~and~Computer~Engineering,~Georgia~Institute~of~Technology, Atlanta,~GA and with GT-CNRS UMI 2958, Metz, France. J. Kliewer is with the Department of Electrical and Computer Engineering, New Jersey Institute of Technology, Newark, NJ. This work was supported in part by NSF grants CCF-1320304 and CCF-1440014.
} \thanks{
E-mail : remi.chou@wichita.edu; matthieu.bloch@ece.gatech.edu; jkliewer@njit.edu. 
Part of the results were presented at the 2015 IEEE International Symposium on Information Theory~\cite{Chou15}.}}
\maketitle
\begin{abstract}
We design polar codes for empirical coordination and strong coordination in two-node networks. \textcolor{black}{ Our  constructions hinge on the fact that polar codes enable explicit low-complexity schemes for soft covering. We leverage this property to propose \emph{explicit}  and low-complexity coding schemes that achieve the capacity regions of both empirical coordination and strong coordination for sequences of actions taking value in an alphabet of prime cardinality. Our results improve previously known polar coding schemes, which (i)~were restricted to uniform distributions and to actions obtained via binary symmetric channels for strong coordination, (ii)~required a non-negligible amount of common randomness for empirical coordination, and (iii) assumed that the simulation of discrete memoryless channels could be perfectly implemented. As a by-product of our results, we obtain a polar coding scheme that achieves channel resolvability for an arbitrary discrete memoryless channel whose input alphabet has prime cardinality.}
\end{abstract}


\section{Introduction}
\label{sec:introduction}

The characterization of the information-theoretic limits of coordination in networks has recently been investigated, for instance, in~\cite{soljanin2002compressing,Cuff10,cuff2013distributed,yassaee2015channel}. The coordinated actions of nodes in a network are modeled by joint probability distributions, and the level of coordination is measured in terms of how well these joint distributions approximate a target joint distribution with respect to the variational distance. Two types of coordination have been introduced: {empirical} coordination, which requires the one-dimensional empirical distribution of a sequence of $n$ actions to approach a target distribution, and strong coordination, which requires  the $n$-dimensional distribution of a sequence of $n$ actions to approach a target distribution. The concept of coordination sheds light into the fundamental limits of several problems, such as distributed control or task assignment in a network.
\textcolor{black}{Several extensions and applications have built upon the results of~\cite{Cuff10}, including channel simulation~\cite{cuff2013distributed,yassaee2015channel}, 
multiterminal settings for empirical coordination~\cite{bereyhi2013empirical} or strong coordination \cite{haddadpour2012coordination,bloch2014strong,vellambi2016strong}, empirical coordination for joint source-channel coding~\cite{treust2014empirical}, and coordination for power control~\cite{larrousse2015coordination}.} \textcolor{black}{Strong coordination also finds its origin in quantum information theory  with ``visible compression of mixed states''~\cite{soljanin2002compressing, dur2001visible, bennett1999entanglement}, \cite[Section 10.6]{hayashi2006quantum}.}


The design of practical and efficient coordination schemes approaching the fundamental limits predicted by information theory has, however, attracted little attention to date. One of the hurdles faced for code design is that the metric to optimize is not a probability of error but a variational distance between distributions. \textcolor{black}{Notable exceptions are~\cite{blasco12,Bloch12s}, which have proposed coding schemes based on polar codes~\cite{Arikan09,Arikan10} for a small subset of all  two-node network coordination problems. 
}

\textcolor{black}{
In this paper, we demonstrate how to solve the issue of coding for channel resolvability with polar codes. Building upon this result, we extend the constructions in\cite{blasco12,Bloch12s} to provide an explicit and low-complexity alternative to the information-theoretic proof in~\cite{Cuff10} for two-node networks. More specifically, the contributions of this paper are as follows:
\begin{itemize}
\item We propose an explicit polar coding scheme to achieve the channel resolvability of an arbitrary memoryless channel whose input alphabet has prime cardinality; low-complexity coding schemes have previously been proposed with polar codes in~\cite{Bloch12s}, invertible extractors in~\cite{chou2014res}, and injective group homomorphisms in \cite{hayashi2012secure} but are all restricted to symmetric channels. \textcolor{black}{Although \cite{amjad2015channel} has proposed  low-complexity linear coding schemes for arbitrary memoryless channels, the construction therein is non-explicit in the sense that only existence results are proved.}
\item We propose an explicit polar coding scheme that achieves the empirical coordination capacity region for actions from an alphabet of prime cardinality, when common randomness, whose rate vanishes to zero as the blocklength grows, is available at the nodes. This construction extends~\cite{blasco12}, which only deals with uniform distributions and requires a non negligible rate of common randomness available at the nodes.
\item We propose an explicit polar coding scheme that achieves the strong coordination capacity region for actions from an alphabet with prime cardinality. This generalizes~\cite{Bloch12s}, which only considers uniform distributions of actions obtained via a binary symmetric channel, and assumes that the simulation of discrete memoryless channels can be perfectly implemented.
\end{itemize}
}
Our proposed constructions are \emph{explicit} and handle  asymmetric settings through block-Markov encoding instead of relying, as in \cite{Honda13} and related works, on the existence of some maps or a non-negligible amount of shared randomness; we provide further discussion contrasting the present work with~\cite{Honda13} in Remark~\ref{remq}.

The coding mechanism underlying our coding schemes is ``soft covering," which refers to the approximation of output statistics using codebooks. In particular, soft covering with random codebooks has been used to study problems in information theory such as Wyner's common information~\cite{Wyner75c}, the resolvability of a channel \cite{Han93}, secrecy over wiretap channels~\cite{Hayashi06,hayashi2012secure,Csiszar1996,bloch2013strong,goldfeld2015semantic}, \textcolor{black} {secrecy over quantum wiretap channels~\cite{hayashi2015quantum}}, strong coordination~\cite{Cuff10}, channel synthesis~\cite{cuff2013distributed,yassaee2015channel}, covert and stealth communication \cite{hou2014effective,bloch2015covert}. 
In our coding schemes, soft covering with polar codes is obtained via the special case of soft covering over noiseless channels, together with an appropriate block-Markov encoding to ``recycle" common randomness.

\begin{rem} \label{remq}
\cite[Theorem 3]{Honda13} provides a polar coding scheme for asymmetric channels for which reliability holds on average over a random choice of the sequence of ``frozen bits." This proves the existence of a specific sequence of ``frozen bits" that ensures reliability. \cite[Section~III-A]{Honda13} also provides an \emph{explicit} construction, which, however, requires that encoder and decoder share a non-negligible amount of randomness as the blocklength grows.
To circumvent these issues, we use instead the technique of block-Markov encoding, which has been successfully applied to universal channel coding in~\cite{hassani2014universal}, to channel coding in \cite{Mondelli14}, and to Wyner-Ziv coding in~\cite{Chou14rev},~\cite{Chou15f}.
\end{rem}

\textcolor{black} {The idea of randomness recycling is closely related to recursive constructions of seeded extractors in the computer science literature, see for instance \cite{vadhan2012pseudorandomness}. Block-Markov encoding in polar coding schemes has first been used in \cite{hassani2014universal,Mondelli14}, for problems involving reliability constraints and requiring the  reconstruction of some random variables.} Unlike problems that only involve reliability constraints, an additional difficulty of block-Markov encoding for coordination is to ensure approximation of  the target distribution jointly over all encoding blocks, despite potential inter-block dependencies.

The remainder of the paper is organized as follows. Section~\ref{sec:notations} provides the notation, and Section~\ref{sec:ps} reviews the notion of resolvability and coordination. Section~\ref{sec:resolvability} demonstrates the ability of polar codes to achieve channel resolvability. For a two-node network, Sections~\ref{sec:EmpCord} and~\ref{sec:StrongCord} provide polar coding schemes that achieve the empirical coordination capacity region and the strong coordination capacity region, respectively. Finally, Section~\ref{secc} provides concluding remarks. 

\section{Notation} \label{sec:notations}
We let $\llbracket a,b \rrbracket$ be the set of integers between $\lfloor a \rfloor$ and~$\lceil b \rceil$. We denote the set of strictly positive natural numbers by~$\mathbb{N}^*$. For $n \in \mathbb{N}$, we let $G_n \triangleq  \left[ \begin{smallmatrix}
       1 & 0            \\[0.3em]
       1 & 1 
     \end{smallmatrix} \right]^{\otimes n} $ be the source polarization transform defined in~\cite{Arikan10}. The components of a vector $X^{1:N}$ of size $N$ are denoted with superscripts, i.e., $X^{1:N} \triangleq (X^1 , X^2, \ldots, X^{N})$. For any set $\mathcal{A} \subset \llbracket 1,N \rrbracket$, we let $X^{1:N}[\mathcal{A}]$ be the components of $X^{1:N}$ whose indices are in $\mathcal{A}$. 
For two probability distributions $p$ and $q$ defined over the same alphabet $\mathcal{X}$, \textcolor{black} {we define the variational distance between $p$ and $q$ as 
$$\mathbb{V}(p_X,q_X) \triangleq \sum_{x\in \mathcal{X}} |p(x) - q(x)| .$$ }
We let $\mathbb{D}(\cdot || \cdot)$ denote the Kullback-Leibler divergence between two distributions. If $p$, $q$ are two distributions over the finite alphabet $\mathcal{X}$, similar to \cite{kramerbook}, we use the convention $D(p||q) = + \infty$ if there exists $x\in\mathcal{X}$ such that $q(x) =0$ and $p(x)>0$. For joint probability distributions $p_{XY}$ and $q_{XY}$ defined over $\mathcal{X}\times \mathcal{Y}$, we write the conditional Kullback-Leibler divergence~as
\begin{align*}
\mathbb{E}_{p_X} \left[ \mathbb{D}(p_{Y|X}||q_{Y|X}) \right]  \triangleq \sum_{x\in \mathcal{X}} p_X(x)\mathbb{D}(p_{Y|X=x}||q_{Y|X=x}).
\end{align*} All $\log$s are taken with respect to base $2$. 
Unless specified otherwise, random variables are written in upper case letters, and particular realizations of a random variable are written in corresponding lower case letters. Finally, the indicator function is denoted by $\mathds{1}\{ \omega \}$, which is equal to $1$ if the predicate $\omega$ is true and $0$ otherwise.

\section{Review of resolvability and coordination} \label{sec:ps}
\textcolor{black}{We first review the notion of channel resolvability~\cite{Han93}, which plays a key role in the analysis of strong coordination. We then review the related notion of resolvability, which forms  a building block of our proposed coding schemes. Finally, we review the definition of empirical and strong coordination in two-node networks as introduced in \cite{Cuff10}. A common characteristic of these three problems is that they do not require a reconstruction algorithm to ensure a reliability condition, as, for instance, for source or channel coding, but require, instead a ``good" approximation of given probability distributions.
}

\subsection{Soft Covering and Channel Resolvability} \label{sec:psresolv}
\begin{figure}
\centering
  \includegraphics[width=5cm]{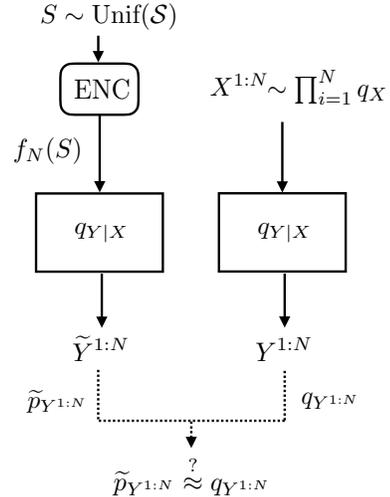}
  \caption{Description of the channel resolvability problem. $\textup{Unif}(\mathcal{S})$ denotes the uniform distribution over $\mathcal{S}$.}
  \label{fig:modelcres}
\end{figure}
\textcolor{black}{ 
 Soft covering is concerned with the approximation theory of output statistics \cite{Han93} and  first appeared in the analysis of the common information between random variables~\cite{Wyner75c}. 
 Consider a discrete memoryless channel $(\mathcal{X}, q_{Y|X}, \mathcal{Y})$ and a memoryless source $(\mathcal{X},q_X)$ with $\mathcal{X}$ and  $\mathcal{Y}$ representing finite alphabets. Define the target distribution $q_Y$ as the output of the channel when the input distribution is $q_X$ as
\begin{align}
\forall y \in \mathcal{Y}, q_Y (y) \triangleq \sum_{x \in \mathcal{X}} q_{Y|X}(y|x) q_X(x). \label{eqcresolvability}
\end{align}
As depicted in Figure \ref{fig:modelcres}, the encoder wishes to form a sequence with the least amount of randomness such that the sequence sent over the channel $(\mathcal{X}, q_{Y|X}, \mathcal{Y})$ produces an output $\widetilde{Y}^{1:N}$, whose distribution is close to $q_{Y^{1:N}} \triangleq \prod_{i=1}^Nq_{Y}$.
A formal definition is given as follows.
\begin{defn}
Consider a discrete memoryless channel $(\mathcal{X}, q_{Y|X}, \mathcal{Y})$. A $(2^{NR},N)$ soft covering code $\mathcal{C}_N$ consists of
\begin{itemize}
\item a randomization sequence $S$ uniformly distributed over $\mathcal{S} \triangleq \llbracket 1, 2^{NR} \rrbracket$;
\item an encoding function $f_N : \mathcal{S} \to \mathcal{X}^N$;
\end{itemize}
and operates as follows:
\begin{itemize}
\item the encoder forms $f_N(S)$;
\item the encoder transmits $f_N(S)$ over the channel $q_{Y^{1:N}|X^{1:N}} = \prod_{i=1}^N q_{Y|X}$.
\end{itemize}
\end{defn}
}
\begin{defn}
$R$ is an achievable soft covering rate for an input distribution $q_X$ if there exists a sequence of $(2^{NR},N)$ soft covering codes, $\{ \mathcal{C}_N \}_{N \geq 1}$, such that
$$
\lim_{N \to \infty } \mathbb{D} (\widetilde{p}_{Y^{1:N}}||q_{Y^{1:N}}) = 0,
$$
where $q_{Y^{1:N}} = \prod_{i=1}^N q_Y$ with $q_Y$ defined in \eqref{eqcresolvability} for the input distribution $q_X$ and $\forall y^{1:N} \in \mathcal{Y}^N,$
\begin{align*}
 \widetilde{p}_{Y^{1:N}} (y^{1:N}) 
 \triangleq \sum_{s \in \mathcal{S}} q_{Y^{1:N}|X^{1:N}}\left(y^{1:N}|f_N(s) \right) \frac{1}{|\mathcal{S}|}.
\end{align*}
\end{defn}
As summarized in Theorem \ref{th_Han}, the infimum of achievable rates for any input distribution $q_X$ is called channel resolvability and has been characterized in \cite{Han93,hayashi2012secure,cuff2013distributed,hou2014effective} for various metrics. 
\begin{thm}[Channel Resolvability] \label{th_Han}
Consider a discrete memoryless channel $W \triangleq (\mathcal{X}, q_{Y|X}, \mathcal{Y})$. The channel resolvability of $W$ is $\displaystyle\max_{q_X}I(X;Y)$, where $X$ and $Y$ have joint distribution $q_{Y|X}q_X$.
\end{thm}
\subsection{Resolvability and Conditional Resolvability} \label{sec:sr}
\textcolor{black}{Resolvability corresponds to channel resolvability over a noiseless channel and thus characterizes the minimum rate of a uniformly distributed sequence required to simulate a source with given statistics. 
We review here the slightly more general notion of \emph{conditional  resolvability}~\cite{steinberg1996simulation}.
}


Consider a discrete memoryless source $(\mathcal{X}\times \mathcal{Y},q_{XY})$ with $\mathcal{X}$ and  $\mathcal{Y}$ finite alphabets. As depicted in Figure \ref{fig:modelsres}, given $N$~realizations of the memoryless source $(\mathcal{Y},q_{Y})$, the encoder wishes to form $\widetilde{X}^{1:N}$ with the minimal amount of randomness such that the joint distribution of $(\widetilde{X}^{1:N},{Y}^{1:N})$, denoted by $\widetilde{p}_{X^{1:N}Y^{1:N}}$, is close to $q_{X^{1:N}Y^{1:N}} \triangleq \prod_{i=1}^Nq_{XY}$. We refer to this setting as the conditional resolvability problem. Note that the traditional definition of resolvability~\cite{Han93} corresponds to $\mathcal{Y} = \emptyset$. A formal definition is as follows.
\begin{figure}
\centering
  \includegraphics[width=8cm]{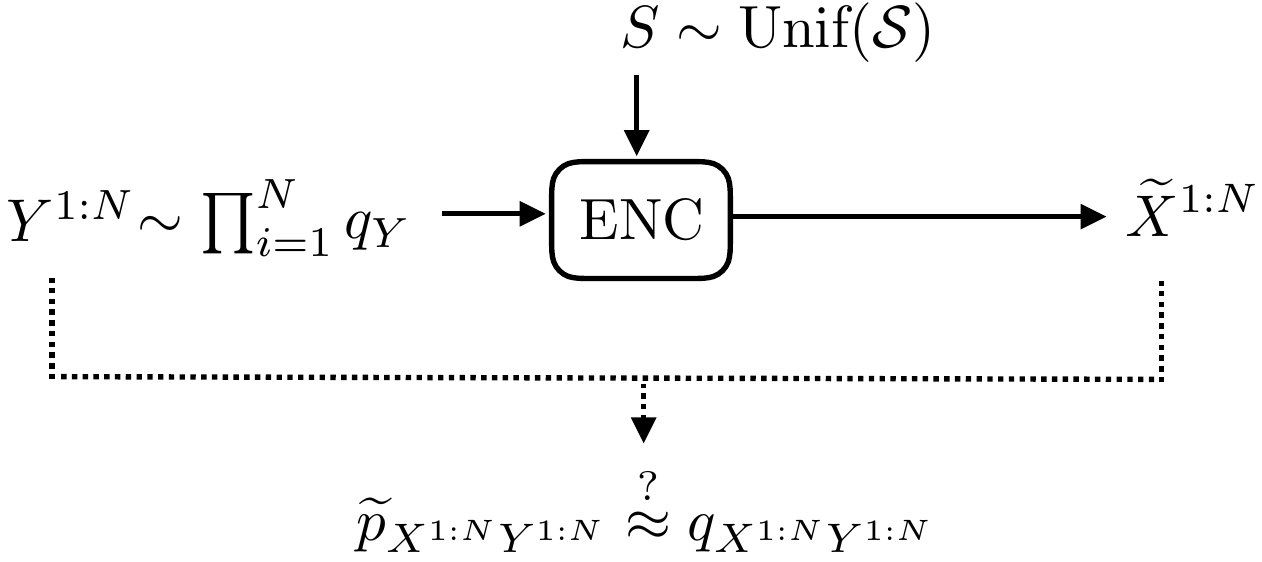}
  \caption{Illustration of the conditional resolvability problem. $\textup{Unif}(\mathcal{S})$ denotes the uniform distribution over $\mathcal{S}$.}
  \label{fig:modelsres}
\end{figure}
\begin{defn}  \label{def2}
A $(2^{NR},N)$ code $\mathcal{C}_N$ for a discrete memoryless source $(\mathcal{X}\times \mathcal{Y},q_{XY})$ consists of
\begin{itemize}
\item a randomization sequence $S$ uniformly distributed over $\mathcal{S}  \triangleq \llbracket 1, 2^{NR} \rrbracket$;
\item an encoding function $f_N : \mathcal{S} \times \mathcal{Y}^N\to \mathcal{X}^N$;
\end{itemize}
and operates as follows:
\begin{itemize}
\item the encoder observes $N$ realizations $Y^{1:N}$ of  the memoryless source $(\mathcal{Y},q_{Y})$;
\item the encoder forms $ \widetilde{X}^{1:N} \triangleq f_N(S,Y^{1:N})$.
\end{itemize}
The joint distribution of $(\widetilde{X}^{1:N},Y^{1:N})$ is denoted by $\widetilde{p}_{X^{1:N}Y^{1:N}}$.
\end{defn}
\begin{defn} \label{def1}
$R$ is an achievable conditional resolution rate for a discrete memoryless source $(\mathcal{X}\times \mathcal{Y},q_{XY})$ if there exists a sequence of $(2^{NR},N)$ codes, $\{ \mathcal{C}_N \}_{N \geq 1}$ such that
$$
\lim_{N \to \infty } \mathbb{D} (\widetilde{p}_{X^{1:N}Y^{1:N}}||q_{X^{1:N}Y^{1:N}}) = 0,
$$
 with $\widetilde{p}_{X^{1:N}Y^{1:N}}$ the joint  probability distribution of the encoder output $\widetilde{X}^{1:N}$, and where $Y^{1:N}$ is available at the encoder.
\end{defn}
%
%
The infimum of such achievable rates is called the conditional resolvability and is characterized as follows~\cite{steinberg1996simulation,cuff2013distributed}.
\begin{thm}[Conditional Resolvability] \label{thsr}
The conditional resolvability of a discrete memoryless source $(\mathcal{X}\times \mathcal{Y},q_{XY})$ is $H(X|Y)$.
\end{thm}

\begin{rem}
In \cite{steinberg1996simulation}, the conditional  resolvability is described as the minimum randomness required to approximate a target conditional distribution representing a channel given a fixed input process. We prefer to approach conditional resolvability as an extension of resolvability since the corresponding interpretation in terms of random number generation \cite{Han93} and the special case $\mathcal{Y}=\emptyset$ seem more natural in the context of our proofs. 
\end{rem}

The operation described in Definitions \ref{def2} and \ref{def1} may be viewed as performing soft covering over noiseless channels. Codes achieving conditional resolvability will be the main building block of our coding schemes to emulate soft covering over noisy channels. 

\subsection{Coordination} \label{secdefcoord}
\begin{figure}
\centering
  \includegraphics[width=7.8cm]{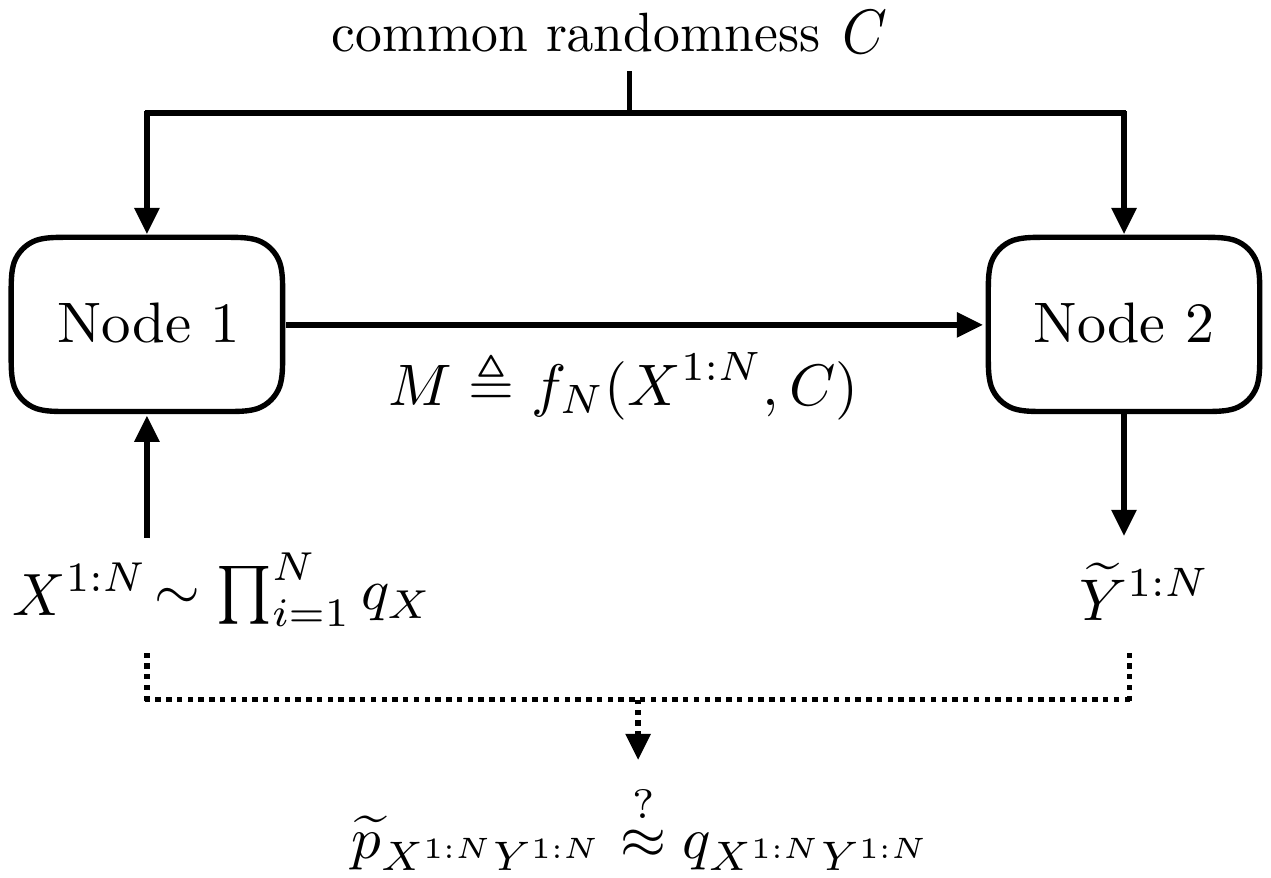}
  \caption{Coordination setup for a two-node network.}
  \label{fig:model}
\end{figure}


\textcolor{black}{Consider a memoryless source $(\mathcal{X} \mathcal{Y}, q_{XY})$  with $\mathcal{X}$ and  $\mathcal{Y}$ finite alphabets, and two nodes, Node~$1$ and Node $2$. \textcolor{black}{As depicted in Figure \ref{fig:model}, Node 1 observes a sequence of actions $X^{1:N}$ and sends a message $M$ over a noiseless channel to Node 2. $M$ must be constructed such that from $M$ and some randomness $C$, pre-shared with Node~1, Node 2 can produce $\widetilde{Y}^{1:N}$ such that the joint distribution of $({X}^{1:N},\widetilde{Y}^{1:N})$, denoted by $\widetilde{p}_{X^{1:N}Y^{1:N}}$, is close to $q_{Y^{1:N}} \triangleq \prod_{i=1}^Nq_{XY}$. A formal definition is as follows.}
 \begin{defn}
A $(2^{NR},2^{NR_0},N)$ coordination code $\mathcal{C}_N$ for a fixed joint distribution $q_{XY}$ consists of
\begin{itemize}
\item common randomness $C$ with rate $R_0$ shared by Node 1 and Node 2;
\item an encoding function $f_N : \mathcal{X}^N \times \llbracket 1 , 2^{NR_0}\rrbracket \to \llbracket 1 , 2^{NR} \rrbracket$ at Node 1;
\item a decoding function $g_N :  \llbracket 1 , 2^{NR} \rrbracket \times \llbracket 1 , 2^{NR_0}\rrbracket \to \mathcal{Y}^N$ at Node 2,
\end{itemize}
and operates as follows:
\begin{itemize}
\item Node $1$ observes $X^{1:N}$, $N$ independent realizations of $(\mathcal{X}, q_{X})$;
\item Node $1$ transmits $f_N(X^{1:N},C)$ to Node $2$;
\item Node $2$ forms $\widetilde{Y}^{1:N} \triangleq g_N (f_N(X^{1:N},C),C)$, whose joint distribution with $X^{1:N}$ is denoted by $\widetilde{p}_{X^{1:N}Y^{1:N}}$.
\end{itemize}
\end{defn}
The notion of empirical and strong coordination are then defined as follows.
}

\begin{defn}
A rate pair $(R,R_0)$ for a fixed joint distribution $q_{XY}$ is achievable for empirical coordination if there exists a sequence of $(2^{NR},2^{NR_0},N)$ coordination codes $\{ \mathcal{C}_N \}_{N \geq 1}$ such that  for $\epsilon >0$
$$
 \lim_{N \to \infty} \mathbb{P} [ \mathbb{V} \left( q_{XY}  , T_{{X}^{1:N}\widetilde{Y}^{1:N}} \right) > \epsilon ] =0,
$$
where for a sequence $({x}^{1:N},\widetilde{y}^{1:N})$ generated at Nodes~1, 2, and for $(x,y) \in \mathcal{X} \times \mathcal{Y} $, 
$$T_{{x}^{1:N}\widetilde{y}^{1:N}}(x,y) \triangleq  \frac{1}{N} \sum_{i=1}^N \mathds{1} \{( {x}^i,\widetilde{y}^i) = (x,y) \},$$  is the joint histogram of $({x}^{1:N},\widetilde{y}^{1:N})$. The closure of the set of achievable rates is called the empirical coordination capacity region.
\end{defn}

\begin{defn}
A rate pair $(R,R_0)$ for a fixed joint distribution $q_{XY}$ is achievable for strong coordination if there exists a sequence of $(2^{NR},2^{NR_0},N)$ coordination codes, $\{ \mathcal{C}_N \}_{N \geq 1}$ such that 
$$
\lim_{N \to \infty}\mathbb{V}(\widetilde{p}_{X^{1:N}Y^{1:N}}, q_{X^{1:N}Y^{1:N}}) = 0.
$$
The closure of the set of achievable rate pairs is called the strong coordination capacity region.
\end{defn}

The capacity regions for empirical coordination and strong coordination have been fully characterized in~\cite{Cuff10}.
\begin{thm}[\!\!\cite{Cuff10}] \label{Th1}
The empirical coordination capacity region is
\vspace*{-1em}
\begin{equation*}
\mathcal{R}_{\textup{EC}}(q_{XY}) \triangleq \{ (R,R_0) : R \geq I(X;Y), \text{ }R_0\geq 0 \}.
\end{equation*}
\end{thm}

\begin{thm}[\!\!\cite{Cuff10}] \label{th_c}
The strong coordination capacity region is
\begin{align*}
&\mathcal{R}_{\textup{SC}}(q_{XY}) \\ 
& \triangleq \!\!\!\!\!\! \bigcup_{ \substack{X \to \ V \to Y \\ |\mathcal{V}|\leq |\mathcal{X}||\mathcal{Y}| +1}} \!\!\! \{ (R,R_0) : R +R_0 \geq I(XY;V), \!\text{ } R \geq I(X;V) \}.
\end{align*}
\end{thm}

%

\section{Polar Coding for Channel Resolvability} \label{sec:resolvability}
\textcolor{black}{We first develop an explicit and low-complexity coding scheme to achieve channel resolvability. The key ideas that will be reused in our coding scheme for empirical and strong coordination are (i)~resolvability-achieving random number generation and (ii) randomness recycling through block-Markov  encoding.} %

Informally, our coding scheme operates over $k \in \mathbb{N}^*$ encoding blocks of length $N \triangleq 2^n$, $n \in \mathbb{N}^*$ as follows. In the first block, using a rate $H(X)$ of randomness, we generate a random variable whose distribution is close to $q_{X^{1:N}}$. When the produced random variable is sent over the channel $q_{Y|X}$, the channel output distribution is close to $q_{Y^{1:N}}$. The amount of randomness used is non-optimal, since we are approximately ``wasting" a fraction $H(X|Y)$ of randomness by Theorem~\ref{th_Han}. For the next encoding blocks, we proceed as in the first block except that part of the randomness is now recycled from the previous block. More specifically, we recycle the bits of randomness used at the input of the channel in the previous block that are almost independent from the channel output. The rate of those bits can be shown to approach $H(X|Y)$. The main difficulty is to ensure that the target distribution at the output of the channel is jointly approximated over all blocks despite the randomness reuse from one block to another.

\textcolor{black}{We provide a formal description of the coding scheme in Section~\ref{sec:CS}, and present its analysis in Section \ref{sec:anresolv}. Part of the analysis for channel resolvability will be directly reused for the problem of strong coordination in Section~\ref{sec:StrongCord}. 
}
\subsection{Coding Scheme} \label{sec:CS}

Fix a joint probability distribution $q_{XY}$ over  $\mathcal{X} \times \mathcal{Y}$, where $|\mathcal{X}|$ is a prime number. Define $U^{1:N} \triangleq X^{1:N} G_n$, where $G_n$ is defined in Section \ref{sec:notations}, and define for $\beta < 1/2$, $\delta_N \triangleq 2^{-N^{\beta}}$ and the sets
\begin{align*}
\mathcal{V}_X &\triangleq \left\{ i \in \llbracket 1, N \rrbracket : H( U^i | U^{1:i-1}) >  \log  |\mathcal{X}| - \delta_N \right\},\\
\mathcal{V}_{X|Y} &\triangleq \left\{ i \in \llbracket 1, N \rrbracket : H( U^i | U^{1:i-1} Y^N) > \log |\mathcal{X}| - \delta_N \right\}.
\end{align*}
Note that the sets $\mathcal{V}_X$ and $\mathcal{V}_{X|Y}$ are defined with respect to $q_{XY}$.
Intuitively, $U^{1:N}[\mathcal{V}_{X|Y}]$ corresponds to the components of $U^{1:N}$ that are almost independent from $Y$  (see \cite{Chou14c} for an interpretation of $\mathcal{V}_X$ and $\mathcal{V}_{X|Y}$ in terms of randomness extraction). Note also that 
\begin{align*}
\lim_{N\to \infty} |\mathcal{V}_{X}| / N = H(X), \\ \lim_{N\to \infty} |\mathcal{V}_{X|Y}| / N = H(X|Y)
\end{align*} by \cite[Lemma 7]{Chou14c}, which relies on a proof technique used in~\cite{Honda13} and a result from \cite{karzand2010polar}.

We use the subscript $i \in \llbracket 1,k \rrbracket$ to denote random variables associated with the encoding of Block~$i$, and we use the notation $X_{i:j} \triangleq (X_l)_{l \in \llbracket i, j \rrbracket}$, when $i<j$. The encoding process is described in Algorithm \ref{alg:encoding_cr}. The functional dependence graph of the coding scheme is depicted in Figure \ref{figFGD} for the reader's convenience.
\begin{figure}
\centering
  \includegraphics[width=8.1cm]{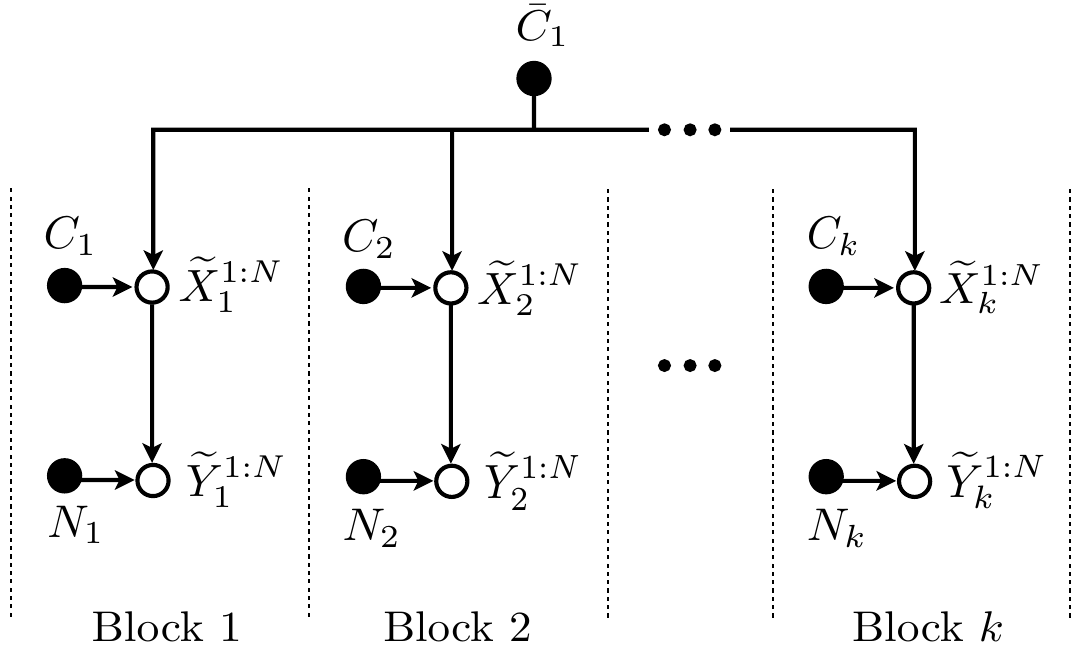}
  \caption{Functional dependence graph of the block encoding scheme for channel resolvability. $N_i$, $i \in \llbracket 1 ,k \rrbracket$, is the channel noise corresponding to the transmission over Block $i$.  For Block $i$, $(C_i, \bar{C}_{i-1}$) is the randomness used at the encoder to form $\widetilde{X}_i$, where $\forall i \in \llbracket 2, k \rrbracket$, $\bar{C}_i = \bar{C}_{i-1}$ and $C_i$ is only used in Block $i$.}
  \label{figFGD}
\end{figure}
\begin{algorithm}[]
  \caption{Encoding algorithm for channel resolvability}
  \label{alg:encoding_cr}
  \begin{algorithmic}   [1] 
    \REQUIRE A vector $\bar{C}_1$ of $|\mathcal{V}_{X|Y}|$ uniformly distributed symbols shared by the encoder and decoder, and $k$ vectors $C_{1:k}$ of $|\mathcal{V}_{X} \backslash \mathcal{V}_{X|Y}|$ uniformly distributed symbols.
    \FOR{Block $i=1$ to $k$}
    \STATE $\bar{C}_i \leftarrow \bar{C}_1$
    \STATE $\widetilde{U}_i^{1:N}[\mathcal{V}_{X|Y}]\leftarrow \bar{C}_i$
    \STATE $\widetilde{U}_i^{1:N}[\mathcal{V}_{X} \backslash \mathcal{V}_{X|Y}] \leftarrow C_i$
    \STATE Successively draw the remaining components $\widetilde{U}_i^{1:N}[\mathcal{V}_{X}^c]$, according to 
      \begin{align}  
&\widetilde{p}_{U_i^j|U_i^{1:j-1}} (u_i^j|\widetilde{U}_i^{1:j-1}) \nonumber \\
&\triangleq 
        {q}_{U^j|U^{1:j-1}} (u_i^j|\widetilde{U}_i^{1:j-1}) \text{ if }  j \in  \mathcal{V}_{X}^c. \label{true_distribcr}
\end{align}
    \STATE Transmit $\widetilde{X}_i^{1:N} \triangleq \widetilde{U}_i^{1:N} G_n$ over the channel $q_{Y|X}$. We denote $\widetilde{Y}_i^{1:N}$ the corresponding channel output.
    \ENDFOR
  \end{algorithmic}
\end{algorithm}
In essence, the protocol described in Algorithm \ref{alg:encoding_cr} performs a resolvability-achieving random number generation~\cite[Definition 2.2.2]{HanBook} for each encoding block, and recycles randomness $\bar{C}_1$ over all blocks.

\begin{rem}
The randomizations described in \eqref{true_distribcr} could be replaced by deterministic decisions for $j \in \mathcal{H}_X^c$, i.e., randomized decisions are only needed for $j \in \mathcal{V}_X^c\backslash \mathcal{H}_X^c$, as shown in \cite{Chou15f}.
\end{rem}

\subsection{Scheme Analysis} \label{sec:anresolv}
Our analysis of Algorithm \ref{alg:encoding_cr} exploits simple relations satisfied by the Kullback-Leibler divergence presented in Appendix~\ref{appdiv}. \textcolor{black}{We denote the distribution induced by the coding scheme, i.e., the joint distribution of $\widetilde{X}_i^{1:N}$ and $\widetilde{Y}_i^{1:N}$, by $\widetilde{p}_{X_i^{1:N}Y_i^{1:N}} = {q}_{Y^{1:N}|X^{1:N}} \widetilde{p}_{X_i^{1:N}}$, $i\in \llbracket 1, k \rrbracket$.} We start with Lemma \ref{lem0} that helps us show in Remark~\ref{remB} that each block individually performs soft covering. 
\begin{lem} \label{lem0}
For block $i\in \llbracket 1, k \rrbracket$, we have 
$$
 \mathbb{D} ( q_{X^{1:N}Y^{1:N}} || \widetilde{p}_{X_i^{1:N}Y_i^{1:N}}) \leq  \delta_N^{(1)},
 $$  
 where $\delta_N^{(1)} \triangleq N \delta_N$.   
\end{lem}

\begin{proof}
For $i \in \llbracket 1 , k \rrbracket$, we have
\begin{align*}
&\mathbb{D} ( q_{X^{1:N}Y^{1:N}} || \widetilde{p}_{X_i^{1:N}Y_i^{1:N}}) \\
& = \mathbb{E}_{q_{X^{1:N}}} \left[\mathbb{D} ( q_{Y^{1:N}|X^{1:N}} || \widetilde{p}_{Y_i^{1:N}|X_i^{1:N}}) \right]  +  \mathbb{D} ( q_{X^{1:N}} || \widetilde{p}_{X_i^{1:N}} ) \\
& \stackrel{(a)}{=} \mathbb{D} ( q_{X^{1:N}} || \widetilde{p}_{X_i^{1:N}}) \\
& \stackrel{(b)}{=}  \mathbb{D} ( q_{U^{1:N}}||\widetilde{p}_{U_i^{1:N}} )  \\
& \stackrel{(c)}{=} \sum_{j =1}^N   \mathbb{E}_{q_{U^{1:j-1}}} \left[ \mathbb{D} ( q_{U^{j}|U^{1:j-1}} ||\widetilde{p}_{U_i^{j}|U_i^{1:j-1}} ) \right]  \\
& \stackrel{(d)}{=} \sum_{j \in  \mathcal{V}_{X} }   \mathbb{E}_{q_{U^{1:j-1}}} \left[\mathbb{D} ( q_{U^{j}|U^{1:j-1}} ||\widetilde{p}_{U_i^{j}|U_i^{1:j-1}}) \right]  \\
& \stackrel{(e)}{=} \sum_{j \in \mathcal{V}_{X} }   (\log |\mathcal{X}|- H(U^{j}|U^{1:j-1}))  \\
& \stackrel{(f)}{\leq} |\mathcal{V}_{X}|   \delta_N  \\
& \leq N \delta_N,
\end{align*}
where $(a)$ holds by definition of $\widetilde{p}_{X_i^{1:N}Y_i^{1:N}}$, $(b)$ holds by invertibility of $G_n$, $(c)$ holds by the chain rule for divergence~\cite{Cover91}, $(d)$ holds by \eqref{true_distribcr}, $(e)$ holds by uniformity of the symbols in positions $\mathcal{V}_{X}$, $(f)$ holds by definition of $\mathcal{V}_{X}$.
\end{proof}
\begin{rem} \label{remB}
The encoding algorithm performs a resolvability-achieving random number generation in each block. A formal proof can be found in Appendix \ref{AppA} and relies on Lemma \ref{lem0} and an asymptotic symmetry result for the Kullback-Leibler divergence. 
\end{rem}
We now establish the asymptotic independence of consecutive blocks with the following two lemmas. 
\begin{lem} \label{lem3}
For $i\in \llbracket 2, k \rrbracket$, the outputs of two consecutive blocks are asymptotically independent; specifically, $$
\mathbb{D} \left( \widetilde{p}_{Y_{i-1:i}^{1:N}\bar{C}_1} || \widetilde{p}_{Y_{i-1}^{1:N}\bar{C}_1} \widetilde{p}_{Y_{i}^{1:N}} \right)  \leq \delta_N^{(2)},
$$
where $\delta_N^{(2)} = O\left(N^{2} \delta_N^{1/2} \right)$. 
\end{lem}
\begin{proof}
Let $i \in \llbracket 1 , k \rrbracket$. We have 
\begin{align}
& H({U}^{1:N}[ \mathcal{V}_{X|Y}] | {Y}^{1:N}) - H(\widetilde{U}_{i}^{1:N}[ \mathcal{V}_{X|Y}] | \widetilde{Y}_{i}^{1:N}) \nonumber \\ \nonumber
& =H({U}^{1:N}[ \mathcal{V}_{X|Y}]  {Y}^{1:N})- H(\widetilde{U}_{i}^{1:N}[ \mathcal{V}_{X|Y}]  \widetilde{Y}_{i}^{1:N}) \\ \nonumber
& \phantom{--} + H( \widetilde{Y}_{i}^{1:N}) - H({Y}^{1:N})  \\ \nonumber
& \stackrel{(a)}{\leq} N D_1 \log (|\mathcal{Y}|/D_1) + N D_2 \log (|\mathcal{X}||\mathcal{Y}|/D_2) \\ \nonumber
& \stackrel{(b)}{\leq} N D_1 \log (|\mathcal{Y}|/D_1) + N D_3 \log (|\mathcal{X}||\mathcal{Y}|/D_3) \\ \nonumber
& \stackrel{(c)}{\leq}   2 N \sqrt{2 \ln 2} \sqrt{\delta_N^{(1)}} \log \left(|\mathcal{X}||\mathcal{Y}|/ \left[ \sqrt{2 \ln 2} \sqrt{\delta_N^{(1)}}\right]\right)      \\ 
& \triangleq \delta_N^{(UY)} , \label{eqaaaa}
\end{align}
where $(a)$ holds for $N$ large enough by Lemma \ref{corent} in Appendix~\ref{appdiv} with $D_1 \triangleq \sqrt{2 \ln 2} \sqrt{\mathbb{D}({q}_{U^{1:N} [ \mathcal{V}_{X|Y}]  {Y}^{1:N}} || \widetilde{p}_{U_{i}^{1:N} [ \mathcal{V}_{X|Y}]  {Y}_{i}^{1:N}})}$ and $D_2\triangleq \sqrt{2 \ln 2} \sqrt{ \mathbb{D}({q}_{  {Y}^{1:N}} || \widetilde{p}_{{Y}_{i}^{1:N}})}$, $(b)$ holds for $N$ large enough because $D_2 \leq D_3$ by the chain rule for relative Kullback-Leibler divergence and invertibility of $G_n$ with $D_3 \triangleq \sqrt{2 \ln 2} \sqrt{\mathbb{D}(q_{X^{1:N}Y^{1:N}} ||\widetilde{p}_{X_i^{1:N}Y_i^{1:N}})}$, $(c)$ holds for $N$ large enough by Lemma \ref{lem0} and because $D_1 \leq D_3$.   
Hence, for $i \in \llbracket 2 , k\rrbracket$,
\begin{align*}
&\mathbb{D} \left( \widetilde{p}_{Y_{i-1:i}^{1:N}\bar{C}_1} || \widetilde{p}_{Y_{i-1}^{1:N}\bar{C}_1} \widetilde{p}_{Y_{i}^{1:N}} \right) \\
& = I(\widetilde{Y}_{i-1}^{1:N} \bar{C}_1;  \widetilde{Y}_{i}^{1:N})\\ 
& = I(\widetilde{Y}_{i}^{1:N} ; \bar{C}_1 ) + I(\widetilde{Y}_{i-1}^{1:N} ; \widetilde{Y}_{i}^{1:N} | \bar{C}_1 ) \\
&  \stackrel{(d)}{=} I(\widetilde{Y}_{i}^{1:N} ; \bar{C}_1 )  \\
&  = I(\widetilde{Y}_{i}^{1:N} ; \widetilde{U}_{i}^{1:N}[ \mathcal{V}_{X|Y}])  \\
&  \stackrel{(e)}{=} |\mathcal{V}_{X|Y} | \log |\mathcal{X}|- H(\widetilde{U}_{i}^{1:N}[ \mathcal{V}_{X|Y}] | \widetilde{Y}_{i}^{1:N})  \\
& \stackrel{(f)}{\leq} |\mathcal{V}_{X|Y} |\log |\mathcal{X}| -  H({U}_{}^{1:N}[ \mathcal{V}_{X|Y}] | {Y}_{}^{1:N}) + \delta_N^{(UY)}\\
&  \stackrel{(g)}{\leq}  |\mathcal{V}_{X|Y} |\log |\mathcal{X}|  -  \sum_{j \in \mathcal{V}_{X|Y} }H(U_{}^j| U_{}^{1:j-1}Y_{}^{1:N}) + \delta_N^{(UY)} \\
& \leq  |\mathcal{V}_{X|Y} |\log |\mathcal{X}| -  |\mathcal{V}_{X|Y} | (\log |\mathcal{X}| - \delta_N)  + \delta_N^{(UY)} \\
& \leq N \delta_N  + \delta_N^{(UY)},
\end{align*}
where $(d)$ holds because $\widetilde{Y}_{i-1}^{1:N} - \bar{C}_1 - \widetilde{Y}_{i}^{1:N}  $, as seen in Figure~\ref{figFGD}, $(e)$ holds by uniformity of $\widetilde{U}_{i}^{1:N}[ \mathcal{V}_{X|Y}]$, $(f)$ holds by \eqref{eqaaaa}, $(g)$ holds because conditioning reduces entropy.
\end{proof}
\begin{lem} \label{lem5}
The outputs of all the blocks are asymptotically independent; specifically, 
$$
\mathbb{D} \left( \widetilde{p}_{Y_{1:k}^{1:N}} \left\lVert\ \prod_{i =1}^k  \widetilde{p}_{Y_{i}^{1:N}} \right. \right) \leq (k-1) \delta_N^{(2)} .
$$
\end{lem}
\begin{proof}
We have
\begin{align*}
&\mathbb{D} \left( \widetilde{p}_{Y_{1:k}^{1:N}} \left\lVert\ \prod_{i =1}^k  \widetilde{p}_{Y_{i}^{1:N}}  \right. \right) \\
& \stackrel{(a)}{=} \sum_{i=2}^{k} I(\widetilde{Y}_i^{1:N} ; \widetilde{Y}^{1:N}_{1:i-1})\\
& \leq \sum_{i=2}^{k} I(\widetilde{Y}_i^{1:N} ; \widetilde{Y}^{1:N}_{1:i-1} \bar{C}_1) \\
& = \sum_{i=2}^{k} \left( I(\widetilde{Y}_i^{1:N} ; \widetilde{Y}^{1:N}_{i-1} \bar{C}_1) + I(\widetilde{Y}_i^{1:N} ; \widetilde{Y}^{1:N}_{1:i-2} | \bar{C}_1 \widetilde{Y}^{1:N}_{i-1}) \right)\\
& \stackrel{(b)}{=} \sum_{i=2}^{k}  I(\widetilde{Y}_i^{1:N} ; \widetilde{Y}^{1:N}_{i-1} \bar{C}_1)  \\
& = \sum_{i=2}^{k} \mathbb{D} \left( \widetilde{p}_{Y_{i-1:i}^{1:N}\bar{C}_1} , \widetilde{p}_{Y_{i-1}^{1:N}\bar{C}_1} \widetilde{p}_{Y_{i}^{1:N}} \right)  \\
& \stackrel{(c)}{\leq} \sum_{i=2}^{k}  \delta_N^{(2)}  \\
& = (k-1)  \delta_N^{(2)} ,
& \end{align*}
where $(a)$ holds by Lemma \ref{lemnew} in Appendix \ref{appdiv}, $(b)$ holds because for any $i \in \llbracket 3,k \rrbracket$, the Markov chain  $\widetilde{Y}_{1:i-2}^{1:N} - \bar{C}_1 \widetilde{Y}_{i-1}^{1:N} - \widetilde{Y}_{i}^{1:N}$ holds as seen in Figure \ref{figFGD}, $(c)$ holds by Lemma~\ref{lem3}.
\end{proof}
We are now ready to show that the target output distribution is jointly approximated over all blocks.
\begin{lem} \label{lemAllb}
We have, 
$$
\mathbb{D} \left( \widetilde{p}_{Y_{1:k}^{1:N}} \left\lVert q_{Y^{1:kN}}  \right. \right) \leq  \delta_N^{(3)},
$$
where $\delta_N^{(3)} = O\left(k^{3/2}N^{2} \delta_N^{1/4} \right)$.
\end{lem}
\begin{proof}
First, observe that
	\begin{align}
		\mathbb{D}\left( \left.  q_{Y^{1:kN}} \right\lVert  \prod_{i =1}^k  \widetilde{p}_{Y_{i}^{1:N}}  \right) 
		& = 		\mathbb{D}\left( \prod_{i=1}^k q_{Y^{1:N}}  \left\lVert  \prod_{i =1}^k  \widetilde{p}_{Y_{i}^{1:N}}  \right. \right)  \nonumber \\ \nonumber
				& = 	\sum_{i =1}^k	\mathbb{D}\left( q_{Y^{1:N}}   \left\lVert    \widetilde{p}_{Y_{i}^{1:N}}  \right. \right) \\
				& \leq k \delta_N^{(1)}, \label{eqre}
	\end{align}
	where the inequality holds by Lemma \ref{lem0}. Then, we have, for $N$ large enough, 
\begin{align*}
	&\mathbb{D} \left( \widetilde{p}_{Y_{1:k}^{1:N}} \Vert q_{Y^{1:kN}} \right)	\\
	&  \stackrel{(a)}{\leq}  \log \left(\frac{ 1} { \mu_{q_{Y}}^{ kN}}  \right)\sqrt{2 \ln2} \left[\sqrt{\mathbb{D} \left( \widetilde{p}_{Y_{1:k}^{1:N}} \left\lVert \prod_{i =1}^k  \widetilde{p}_{Y_{i}^{1:N}} \right. \right)}  \right. \\
	& \phantom{-----------}+ \left.\sqrt{\mathbb{D}\left( \left.  q_{Y^{1:kN}} \right\lVert  \prod_{i =1}^k  \widetilde{p}_{Y_{i}^{1:N}}  \right)}  \right]\\
	& \stackrel{(b)}{\leq}  kN \log \left(\frac{ 1} { \mu_{q_{Y}}}  \right)\sqrt{2 \ln2} \left[ \sqrt{k \delta_N^{(1)}} +\sqrt{(k-1)\delta_N^{(2)}}\right], 
	\end{align*}
where $(a)$ holds by Lemma \ref{lemnew2} in Appendix \ref{appdiv} and because $\mu_{q_{Y^{1:kN}}} = \mu_{q_{Y}}^{kN}$, $(b)$ holds by Lemma \ref{lem5} and by \eqref{eqre}. 
\end{proof}
Our encoding scheme exploits randomness to draw symbols according to \eqref{true_distribcr}, whose rate is for any $i \in \llbracket 1 , k \rrbracket$,
  \begin{align*} 
  \lim_{N \to \infty} \frac{1}{N}\sum_{j \in \mathcal{V}_{X}^c} H(\widetilde{U}^{j}_i|\widetilde{U}^{1:j-1}_i).
 \end{align*}
 We quantify this rate in the following lemma. 
\begin{lem} \label{lemint}
	\begin{align} \label{eqrateres}
  \lim_{N \to \infty} \frac{1}{N}\sum_{j \in \mathcal{V}_{X}^c} H(\widetilde{U}^{j}_i|\widetilde{U}^{1:j-1}_i)=0.
 \end{align}
\end{lem}
\begin{proof}
We have for $i \in \llbracket 1, k \rrbracket$, for $j \in \mathcal{V}_{X}^c$,
\begin{align}
&H(\widetilde{U}^{j}_i|\widetilde{U}^{1:j-1}_i) -  H({U}^{j}|{U}^{1:j-1})  \nonumber \\ \nonumber
& = H(\widetilde{U}^{1:j}_i) -  H({U}^{1:j}) + H({U}^{1:j-1}) - H(\widetilde{U}^{1:j-1}_i)   \\ \nonumber
& \stackrel{(a)}{\leq} 2 N D \log (|\mathcal{X}|/D)  \\ \nonumber
& \stackrel{(b)}{\leq}   2 N \sqrt{2 \ln 2} \sqrt{\delta_N^{(1)}} \log \left(|\mathcal{X}|/ \left[ \sqrt{2 \ln 2} \sqrt{\delta_N^{(1)}}\right]\right)      \\ 
& \triangleq \delta_N^{(U)}, \label{eq11a}
\end{align}
where $(a)$ holds by Lemma \ref{corent} in Appendix \ref{appdiv} for $N$ large enough and similar to the proof of Lemma \ref{lem3} with  $D \triangleq \sqrt{2 \ln 2} \sqrt{\mathbb{D}(q_{X^{1:N}} ||\widetilde{p}_{X_i^{1:N}})}$, $(b)$ holds by Lemma~\ref{lem0} for $N$ large enough. 

Defining  $\mathcal{H}_X \triangleq \left\{ i \in \llbracket 1, N \rrbracket : H( U^i | U^{1:i-1}) >   \delta_N \right\}$, we thus obtain 
\begin{align}
& \sum_{j \in \mathcal{V}_{X}^c} H(\widetilde{U}^{j}_i|\widetilde{U}^{1:j-1}_i) \nonumber \\
& \stackrel{}{=} \sum_{j \in \mathcal{H}_{X}^c \cup ( \mathcal{H}_{X} \backslash \mathcal{V}_{X})} H(\widetilde{U}^{j}_i|\widetilde{U}^{1:j-1}_i) \nonumber \\ \nonumber
& \leq |\mathcal{H}_{X} \backslash \mathcal{V}_{X}| \log |\mathcal{X}| + \sum_{j \in \mathcal{H}_{X}^c}  H(\widetilde{U}^{j}_i|\widetilde{U}^{1:j-1}_i)\\ \nonumber
& = (|\mathcal{H}_{X} | -| \mathcal{V}_{X}|)\log |\mathcal{X}| + \sum_{j \in \mathcal{H}_{X}^c}  H(\widetilde{U}^{j}_i|\widetilde{U}^{1:j-1}_i)\\ \nonumber
&  \stackrel{(a)}{\leq} (|\mathcal{H}_{X} | -| \mathcal{V}_{X}|)\log |\mathcal{X}| + \sum_{j \in \mathcal{H}_{X}^c}  (H({U}^{j}|{U}^{1:j-1}) + \delta_N^{(U)})\\ \nonumber
& \stackrel{(b)}{\leq} (|\mathcal{H}_{X} | -| \mathcal{V}_{X}|)\log |\mathcal{X}| + |\mathcal{H}_{X}^c| (\delta_N + \delta_N^{(U)})\\
& \leq (|\mathcal{H}_{X} | -| \mathcal{V}_{X}|)\log |\mathcal{X}| + N (\delta_N + \delta_N^{(U)}), \label{eq11}
\end{align}
where $(a)$ holds by \eqref{eq11a}, $(b)$ holds by definition of $\mathcal{H}_X$.
Hence,~\eqref{eq11} yields \eqref{eqrateres} by \cite[Lemmas 6,7]{Chou14c}.
\end{proof}

Finally, combining the previous lemmas we obtain the following theorem.
\textcolor{black} {
\begin{thm} \label{ThChr}
The coding scheme of Section \ref{sec:CS}, which operates over $k$ blocks of length $N$, achieves channel resolvability with respect to the Kullback-Leibler divergence over the discrete memoryless channel $(\mathcal{X}, q_{Y|X}, \mathcal{Y})$, where $|\mathcal{X}|$ is a prime number.
More specifically, the channel output distribution $\widetilde{p}_{Y_{1:k}^{1:N}}$ approaches the target distribution $q_{Y^{1:kN}}$ with convergence rate
$$
\mathbb{D} \left( \widetilde{p}_{Y_{1:k}^{1:N}} \Vert q_{Y^{1:kN}} \right) = O\left(k^{3/2}N^{2} 2^{-\frac{N^{-\beta}}{4}} \right), \beta \in ]0,1/2[,
$$
with a rate of input of randomness equal to
$
I(X;Y) +\frac{H(X|Y)}{k},
$
as $N \to \infty$, which approaches $I(X;Y)$ as $k \to \infty$. 
 It thus provides an explicit coding scheme with complexity in $O(kN\log N)$ for Theorem~\ref{th_Han}. 
\end{thm}
}
\begin{proof}
By Lemma \ref{lemint}, the overall rate of uniform symbols required is only
\begin{align*}
  \frac{|\bar{C}_1 | + |C_{1:k}|}{kN}
 & = \frac{ |\mathcal{V}_{X|Y}| + k |\mathcal{V}_X \backslash \mathcal{V}_{X|Y}| }{kN} \\
 & = \frac{ |\mathcal{V}_{X|Y}| }{kN} + \frac{|\mathcal{V}_X | - | \mathcal{V}_{X|Y}| }{N} \\
 & \xrightarrow{N \to \infty} I(X;Y) + \frac{H(X|Y)}{k} \\
 &  \xrightarrow{k \to \infty} I(X;Y) ,
 \end{align*}
 where we have used \cite[Lemma 7]{Chou14c}.  Finally, we conclude that the optimal rate $I(X;Y)$ is achieved with Lemma \ref{lemAllb}.  
\end{proof}
\section{Polar Coding for Empirical Coordination} \label{sec:EmpCord}

\textcolor{black}{We now develop an explicit and low-complexity coding scheme for empirical coordination that achieves the capacity region when the actions of Node $2$ are from an alphabet of prime cardinality. The idea is to perform (i) a \emph{conditional} resolvability-achieving random number generation and (ii) randomness recycling through block-Markov encoding. However, the coding scheme is simpler than in Section \ref{sec:resolvability} as the common randomness recycling can be performed and studied more directly. In particular, we will see that the encoding blocks may be treated independently of each other since the approximation of the target distribution is concerned with a one dimensional probability distribution, as opposed to a $kN$ dimensional probability distribution as in Section~\ref{sec:resolvability}.} 

More specifically, the coding scheme can be informally  summarized as follows. From $X^{1:N}$ and some common randomness  of rate close to $H(Y|X)$ shared with Node $2$, Node~$1$ constructs a random variable $\widetilde{Y}^{1:N}$ whose joint probability distribution with $X^{1:N}$ is close to the target distribution $q_{X^{1:N}Y^{1:N}}$. 
Moreover, Node $1$ constructs a message with rate close to $I(X;Y)$ such that Node~$2$ reconstructs $\widetilde{Y}^{1:N}$ with the message and the common randomness. Finally, encoding is performed over $k \in \mathbb{N}^*$ blocks by recycling the same common randomness, so that the overall rate of shared randomness vanishes as the number of blocks increases. We formally describe the coding scheme in Section \ref{sec:CSEC}, and present its analysis in Section~\ref{sec:anEC}.

\subsection{Coding Scheme} \label{sec:CSEC}
\textcolor{black}{In the following, we redefine the following notation to simplify discussion. Consider the random variables $X$, $Y$ distributed according to $q_{XY}$ over $\mathcal{X} \times \mathcal{Y}$, where $|\mathcal{Y}|$ is a prime number.} Let $N \triangleq 2^n$, $n \in \mathbb{N}^*$. Define $U^{1:N} \triangleq Y^{1:N} G_n$, where $G_n$ is defined in Section \ref{sec:notations}, and define for $\beta < 1/2$, $\delta_N \triangleq 2^{-N^{\beta}}$ as well as the sets
\begin{align*}
\mathcal{V}_Y &\triangleq \left\{ i \in \llbracket 1, N \rrbracket : H( U^i | U^{1:i-1}) >  \log |\mathcal{Y}|-\delta_N \right\},\\
\mathcal{H}_Y &\triangleq \left\{ i \in \llbracket 1, N \rrbracket : H( U^i | U^{1:i-1}) >  \delta_N \right\},\\
\mathcal{V}_{Y|X} &\triangleq \left\{ i \in \llbracket 1, N \rrbracket : H( U^i | U^{1:i-1} X^{1:N}) >  \log |\mathcal{Y}| -\delta_N \right\},\\
\mathcal{H}_{Y|X} &\triangleq \left\{ i \in \llbracket 1, N \rrbracket : H( U^i | U^{1:i-1} X^{1:N}) >  \delta_N \right\}.
\end{align*}
Note that the sets $\mathcal{V}_Y$, $\mathcal{H}_Y$, $\mathcal{V}_{Y|X}$, and $\mathcal{H}_{Y|X}$ are defined with respect to $q_{XY}$. Note also that 
\begin{align*}
\lim_{N\to \infty} |\mathcal{V}_{Y}| / N = H(Y) &= \lim_{N\to \infty} |\mathcal{V}_{Y}| / N, \\
\lim_{N\to \infty} |\mathcal{V}_{X|Y}| / N = H(X|Y)&= \lim_{N\to \infty} |\mathcal{H}_{X|Y}| / N,
\end{align*}
 by \cite[Lemmas 6,7]{Chou14c}, where \cite[Lemma 6]{Chou14c} follows from~\cite[Theorem 3.2]{Sasoglu11}. An interpretation of $\mathcal{H}_{Y|X}$ in terms of  source coding with side information is provided in \cite{Arikan10}, and an interpretation of $\mathcal{V}_{Y|X}$  in terms of  privacy amplification is provided in~\cite{Chou14rev,Chou14c}. Encoding is performed over $k \in \mathbb{N}^*$ blocks of length $N$. We use the subscript $i \in \llbracket 1,k \rrbracket$ to denote random variables associated with encoding Block $i$. The encoding and decoding procedures are described in Algorithms~\ref{alg:encoding_2} and \ref{alg:encoding_3}, respectively. 


\begin{rem}
The coding scheme for each block is similar to lossy source coding schemes \cite{Korada10,Honda13}, as suggested by the optimal communication rate described in Theorem~\ref{Th1}. However, the performance metric of interest is different.
\end{rem}

\begin{algorithm}[]
  \caption{Encoding algorithm at Node $1$ for empirical coordination}
  \label{alg:encoding_2}
  \begin{algorithmic}   [1] 
    \REQUIRE A vector $C_1$ of $|\mathcal{V}_{Y|X}|$ uniformly distributed symbols shared with Node $2$ and $X_{1:k}^{1:N}$.
    \FOR{Block $i=1$ to $k$}
    \STATE $C_i \leftarrow C_1$
    \STATE $\widetilde{U}_i^{1:N}[\mathcal{V}_{Y|X}]\leftarrow C_i$
    \STATE Given $X_i^{1:N}$, successively draw the remaining components of $\widetilde{U}_i^{1:N}$ according to $\widetilde{p}_{U_i^{1:N}X_i^{1:N}}$ defined by
  \begin{align}  
& \widetilde{p}_{U_i^j|U_i^{1:j-1}X_i^{1:N}} (u_i^j|\widetilde{U}_i^{1:j-1}X_i^{1:N})  \nonumber \\
&\!\! \triangleq \!\begin{cases}
    {q}_{U^j|U^{1:j-1}X^{1:N}} (u_i^j|\widetilde{U}_i^{1:j-1}X_i^{1:N})  &\!\! \text{if } j \in \mathcal{H}_Y \backslash {\mathcal{V}}_{Y|X}\\
        {q}_{U^j|U^{1:j-1}} (u_i^j|\widetilde{U}_i^{1:j-1}) &\!\! \text{if }  j \in  \mathcal{H}_{Y}^c \label{true_distrib}
 \end{cases} 
\end{align}
    \STATE Transmit $M_i \triangleq \widetilde{U}_i^{1:N}[ \mathcal{H}_Y \backslash {\mathcal{V}}_{Y|X}]$ and $\widetilde{C}_i$, the randomness necessary to draw $\widetilde{U}_i^{1:N}[\mathcal{H}_{Y}^c]$ 
    \ENDFOR
  \end{algorithmic}
\end{algorithm}

\begin{algorithm}[]
  \caption{Decoding algorithm at Node $2$ for empirical coordination}
  \label{alg:encoding_3}
  \begin{algorithmic} [1]   
    \REQUIRE The vectors $C_1$, $\widetilde{C}_{1:k}$, used in Algorithm \ref{alg:encoding_2} and $M_{1:k}$.
    \FOR{Block $i=1$ to $k$}
    \STATE $C_i \leftarrow C_1$
    \STATE $\widetilde{U}_i^{1:N}[\mathcal{V}_{Y|X}]\leftarrow C_i$  
    \STATE $ \widetilde{U}_i^{1:N}[ \mathcal{H}_Y \backslash {\mathcal{V}}_{Y|X}]
\leftarrow M_i$
    \STATE Using $\widetilde{C}_i$, successively draw the remaining components of $\widetilde{U}_i^{1:N}$ according to ${q}_{U^j|U^{1:j-1}} $ 
    \STATE $\widetilde{Y}_i^{1:N} \leftarrow  \widetilde{U}_i^{1:N} G_n$
        \ENDFOR
  \end{algorithmic}
\end{algorithm}

\subsection{Scheme Analysis} \label{sec:anEC}
The following lemma shows that $\widetilde{p}_{X^{1:N}Y^{1:N}}$, defined by $\widetilde{p}_{X^{1:N}} \triangleq q_{X^{1:N}}$ and Equation (\ref{true_distrib}), approximates $q_{X^{1:N}Y^{1:N}}$.
  
\begin{lem} \label{lemstrongcoord}
 For any $ i \in \llbracket 1,k\rrbracket$, 
 $$\mathbb{V}(q_{X^{1:N}Y^{1:N}} , \widetilde{p}_{X_i^{1:N}Y_i^{1:N}}) \leq \sqrt{2 \log 2} \sqrt{N \delta_N}.$$
\end{lem}
\begin{rem}
The encoder in each block performs a conditional resolvability-achieving random number generation by Lemma~\ref{lemstrongcoord} and because $|C_1|/N=|\mathcal{V}_{Y|X}| / N 
 \xrightarrow{N \to \infty} {H(Y|X)}$. \end{rem}
The proof of Lemma \ref{lemstrongcoord} is similar to the proof of Lemma~\ref{lemstrongcoord2} in Section \ref{sec:anSC} and is thus omitted.
The following lemma shows that empirical coordination holds for each block.
\textcolor{black} {
\begin{lem} \label{lem_EC}
Let $\epsilon>0$. For $i\in \llbracket 1 , k \rrbracket$, we have 
$$ 
\mathbb{P} \left[ \mathbb{V} \left( q_{XY} , T_{X^{1:N}_{i}\widetilde{Y}^{1:N}_{i}} \right) > \epsilon \right] \leq \delta^{(A)}(N),
$$
where $ \delta^{(A)}(N ) = O \left( \sqrt{N \delta_N} \right) $.
\end{lem}
\begin{proof}
For $\epsilon >0$, define $\mathcal{T}_{\epsilon}(q_{XY}) \triangleq \{ (x^{1:N},y^{1:N}) : \mathbb{V} \left( q_{XY} , T_{x^{1:N}y^{1:N}} \right) \leq \epsilon \}$. We define for a joint distribution $q$ over $(\mathcal{X}\times \mathcal{Y})$, 
\begin{align*}
&\mathbb{P}_q [(X^{1:N},Y^{1:N}) \in \mathcal{T}_{\epsilon}(q_{XY})] \\ 
&\triangleq  \!\!\!\sum_{x^{1:N},y^{1:N}} \!\!\! q_{X^{1:N}Y^{1:N}} (x^{1:N},y^{1:N}) \mathds{1} \{ (x^{1:N},y^{1:N}) \in \mathcal{T}_{\epsilon}(q_{XY}) \}.
\end{align*}
Note that $\lim_{N \to \infty} \mathbb{P}_q [(X^{1:N},Y^{1:N}) \notin \mathcal{T}_{\epsilon}(q_{XY})] = 0$ by the AEP \cite{Cover91}, and we can precise the convergence rate as follows.
\begin{align}
&	\mathbb{P}_q [(X^{1:N},Y^{1:N}) \notin \mathcal{T}_{\epsilon}(q_{XY})] \nonumber \\ \nonumber
& = 	\mathbb{P}_q \left[\mathbb{V} \left( q_{XY}, T_{X^{1:N} Y^{1:N}} \right) > \epsilon\right]\\ \nonumber
& = 	\mathbb{P}_q \left[\sum_{x,y} | q_{XY} (x,y) - T_{X^{1:N}{Y}^{1:N}}(x,y) | > \epsilon\right]\\ \nonumber
& \leq 	\mathbb{P}_q \left[\exists (x,y), | q_{XY} (x,y) - T_{X^{1:N}{Y}^{1:N}}(x,y) | \geq \frac{\epsilon} {|\mathcal{X}| |\mathcal{Y}|}\right]\\ \nonumber
& \leq 	\sum_{x,y} \mathbb{P}_q \left[| q_{XY} (x,y) - T_{X^{1:N}{Y}^{1:N}}(x,y) | \geq \frac{\epsilon} {|\mathcal{X}| |\mathcal{Y}|}\right]\\ \nonumber
& = 	{\sum_{x,y} \mathbb{P}_q \left[ \left|\frac{1}{N}\sum_{j=1}^N  \left[q_{XY} (x,y) - \mathds{1} \{ (X^{j},{Y}^{j}) = (x,y) \} \right] \right| \right. }\\ \nonumber 
& \left. \phantom{------------------} \geq \frac{\epsilon} {|\mathcal{X}| |\mathcal{Y}|}\right]\\ \nonumber
& \leq 	\sum_{x,y} 2 \exp \left(- \frac{N \epsilon^2}{2 |\mathcal{X}|^2 |\mathcal{Y}|^2}\right)\\
& = 	2 |\mathcal{X}| |\mathcal{Y}| \exp \left(- \frac{N \epsilon^2}{2 |\mathcal{X}|^2 |\mathcal{Y}|^2}\right), \label{eqhoef}
\end{align} 
where the second inequality holds by Hoeffding's inequality applied for each pair $(x,y)$ to the independent and zero-mean random variables $Z_j(x,y) \triangleq \left(q_{XY} (x,y) - \mathds{1} \{ (X^{j},{Y}^{j}) = (x,y)\} \right) \in [-1,1]$, $j \in \llbracket 1 , N \rrbracket$. Next, let $i\in \llbracket 1 , k \rrbracket$. we have 
\begin{align*}
& \mathbb{P}_{\widetilde{p}} \left[ \mathbb{V} \left( q_{XY} ,T_{X^{1:N}_{i}\widetilde{Y}^{1:N}_{i}} \right) > \epsilon \right]  \\
 & = \!\!\!\sum_{x^{1:N},y^{1:N}} \!\!\! \widetilde{p}_{X_i^{1:N}Y_i^{1:N}} (x^{1:N},y^{1:N}) \mathds{1} \{ (x^{1:N},y^{1:N}) \notin \mathcal{T}_{\epsilon}(q_{XY}) \}\\
  & = \!\!\! \sum_{x^{1:N},y^{1:N}} \!\!\! \left[ \widetilde{p}_{X_i^{1:N}Y_i^{1:N}}(x^{1:N},y^{1:N}) -{q}_{X^{1:N}Y^{1:N}} (x^{1:N},y^{1:N})  \right.  \\
 & \phantom{mm}+ \left.  {q}_{X^{1:N}Y^{1:N}} (x^{1:N},y^{1:N}) \right] \mathds{1} \{ (x^{1:N},y^{1:N}) \notin \mathcal{T}_{\epsilon}(q_{XY}) \}  \\
 & \leq  \mathbb{V}(\widetilde{p}_{X_i^{1:N}Y_i^{1:N}} ,{q}_{X^{1:N}Y^{1:N}}) + \mathbb{P}_q [(X^{1:N},Y^{1:N}) \notin \mathcal{T}_{\epsilon}(q_{XY})] \\
 & \leq  \sqrt{2 \log 2} \sqrt{N \delta_N}  + 2 |\mathcal{X}| |\mathcal{Y}| \exp \left(- \frac{N \epsilon^2}{2 |\mathcal{X}|^2 |\mathcal{Y}|^2}\right) ,
\end{align*}
where we have used Lemma \ref{lemstrongcoord} and \eqref{eqhoef}.
\end{proof}
}

\textcolor{black} {
We now show that empirical coordination holds for all blocks jointly.
\begin{lem} \label{lem_ECA}
Let $\epsilon>0$ and $\alpha \in ]0,1/2[$. We have 
$$ 
 \mathbb{P} \left[ \mathbb{V} \left( q_{XY} , T_{X^{1:N}_{1:k}\widetilde{Y}^{1:N}_{1:k}} \right) > \epsilon \right] \leq \delta^{(B)}(N),
$$
where $ \delta^{(B)}(N) =  O\left(k \sqrt{N \delta_N} \right)$.
\end{lem}
\begin{proof}
We have 
\begin{align}
&\mathbb{V} \left( q_{XY} , T_{x^{1:N}_{1:k}\widetilde{y}^{1:N}_{1:k}} \right) \nonumber \\
& = \sum_{x,y} \left| q_{XY}(x,y) - \frac{1}{kN} \sum_{j=1}^{k} \sum_{i=1}^{N} \mathds{1} \{ (x_j^i,\widetilde{y}_j^i)= (x,y) \}  \right|  \nonumber \\ \nonumber
& = \sum_{x,y} \left| \sum_{j=1}^k \left(\frac{1}{k}q_{XY}(x,y) - \frac{1}{kN} \sum_{i=1}^{N} \mathds{1} \{ (x_j^i,\widetilde{y}_j^i)= (x,y) \} \right)  \right|  \\ \nonumber
& \leq \frac{1}{k} \sum_{j=1}^k \sum_{x,y} \left|  q_{XY}(x,y) - \frac{1}{N} \sum_{i=1}^{N} \mathds{1} \{ (x_j^i,\widetilde{y}_j^i)= (x,y) \}   \right| \\
& \leq \frac{1}{k} \sum_{j=1}^k \mathbb{V} \left( q_{XY}  , T_{x_j^{1:N}\widetilde{y}_j^{1:N}}  \right), \label{eqineq}
\end{align}
hence,
\begin{align*}
&\mathbb{P} \left[ \mathbb{V} \left( q_{XY} , T_{X^{1:N}_{1:k}\widetilde{Y}^{1:N}_{1:k}}  \right) > \epsilon \right] \\
& \stackrel{(a)}{\leq} \mathbb{P} \left[ \frac{1}{k} \sum_{j=1}^k \mathbb{V} \left( q_{XY}  , T_{x_j^{1:N}\widetilde{y}_j^{1:N}}  \right) > \epsilon \right] \\
& \leq  \mathbb{P} \left[ \exists j \in \llbracket 1 ,k \rrbracket ,  \mathbb{V} \left( q_{XY}  , T_{x_j^{1:N}\widetilde{y}_j^{1:N}}  \right) > \epsilon \right] \\
& \leq \sum_{j=1}^k \mathbb{P} \left[ \mathbb{V} \left( q_{XY}  , T_{x_j^{1:N}\widetilde{y}_j^{1:N}}  \right) > \epsilon \right] \\
& \stackrel{(b)}{\leq} k \delta_N^{(A)} ,
\end{align*}
where $(a)$ holds by \eqref{eqineq}, $(b)$ holds by Lemma \ref{lem_EC}.
\end{proof}
}%
\begin{thm} \label{ThEC}
\textcolor{black}{The coding scheme described in Algorithms \ref{alg:encoding_2} and \ref{alg:encoding_3}, which operates over $k$ blocks of length $N$, achieves the two-node network empirical coordination capacity region of Theorem \ref{Th1} for an arbitrary target distribution $q_{XY}$ over $\mathcal{X} \times \mathcal{Y}$, where $|\mathcal{Y}|$ is a prime number. The coding scheme is explicit with complexity in $O(kN \log N)$.}
\end{thm}
\begin{proof}
The communication rate is
\begin{align*}
&\frac{k|\mathcal{H}_Y \backslash {\mathcal{V}}_{Y|X}|}{kN} \\
& = \frac{ |{\mathcal{V}}_Y \backslash {\mathcal{V}}_{Y|X}| + |(\mathcal{H}_Y \backslash {\mathcal{V}}_Y  ) \backslash {\mathcal{V}}_{Y|X}|}{N}  \\
& \leq \frac{|{\mathcal{V}}_Y \backslash {\mathcal{V}}_{Y|X}| + |\mathcal{H}_Y \backslash {\mathcal{V}}_Y  |}{N}  \\
& = \frac{|{\mathcal{V}}_Y |-| {\mathcal{V}}_{Y|X}| + |\mathcal{H}_Y |-| {\mathcal{V}}_Y  |}{N}   \xrightarrow{N \to \infty} I(X;Y),
\end{align*}
where we have used \cite[Lemmas 6,7]{Chou14c} for the limit. 
Node 1 also communicates the randomness $\widetilde{C}_{1: k}$ to allow Node 2 to form $\widetilde{U}^{1:N}_{1:k}[\mathcal{H}^c_Y]$, but its rate is $o(N)$ since $$\lim_{N \to \infty} \frac{1}{N}\sum_{j \in \mathcal{H}_{Y}^c} H(\widetilde{U}^{j}|\widetilde{U}^{1:j-1}) = 0,$$
which can be shown using Lemma \ref{lemstrongcoord} similar to the proof of Theorem \ref{ThChr}.
Then, the common randomness rate is
\begin{align*}
\frac{|\mathcal{V}_{Y|X}| }{kN} 
 \xrightarrow{N \to \infty} \frac{H(Y|X)}{k} \xrightarrow{k \to \infty} 0,
\end{align*}
where the limit holds by \cite[Lemma 7]{Chou14c}. 
Finally, we conclude that the region described in Theorem \ref{Th1} is achieved with Lemma \ref{lem_ECA}.
\end{proof}
\section{Polar Coding for Strong Coordination} \label{sec:StrongCord}

\textcolor{black}{We finally design an explicit and low-complexity coding scheme for strong coordination that achieves the capacity region when the actions of Node $2$ are from an alphabet of prime cardinality. The idea is again to perform (i) conditional resolvability-achieving random number generation and (ii) common randomness recycling with block-Markov encoding as in Section \ref{sec:resolvability}. In addition, we also simulate discrete memoryless channels with polar codes as opposed to assuming that this operation can be perfectly realized. 
}

\textcolor{black}{An informal description of the strong coordination coding scheme is as follows. From $X^{1:N}$ and some common randomness of rate close to $H(V|X)$ shared with Node $2$, Node $1$ constructs a random variable $\widetilde{V}^{1:N}$ whose joint probability distribution with $X^{1:N}$ is close to the target distribution $q_{X^{1:N}V^{1:N}}$. Moreover, Node~$1$ constructs a message with rate close to $I(X;V)$ such that Node~$2$ reconstructs $\widetilde{V}^{1:N}$ with the message and the common randomness. Then, Node $2$ simulates a discrete memoryless channel with input $\widetilde{V}^{1:N}$ to form $\widetilde{Y}^{1:N}$ whose joint distribution with $X^{1:N}$ is close to $q_{X^{1:N}Y^{1:N}}$. Finally, \emph{part} of the common randomness is recycled by encoding over $k \in \mathbb{N}^*$ blocks, so that the overall rate of shared randomness becomes on the order of $I(V;Y|X)$.
 As in Section~\ref{sec:resolvability} for channel resolvability, the main difficulty is to ensure that the joint probability distributions of the actions approach the target distribution over all blocks jointly, despite reusing part of the common randomness over all blocks and despite an imperfect simulation of discrete memoryless channels.
}
The coding scheme is formally described in Section~\ref{sec:CSSC}, and its analysis is presented in Section~\ref{sec:anSC}.

\subsection{Coding Scheme} \label{sec:CSSC}

\textcolor{black}{We first redefine the following notation}.
\textcolor{black}{Consider the random variables $X$, $Y$, $V$ distributed according to  $q_{XYV}$ over $\mathcal{X} \times \mathcal{Y} \times\mathcal{V} $ such that $X \to V \to Y$. Moreover, assume that $|\mathcal{Y}|$ and $|\mathcal{V}|$ are prime numbers. By Theorem \ref{th_c}, one can choose $|\mathcal{V}|$ as the smallest prime number greater than or equal to $(|\mathcal{X}||\mathcal{Y}|+1)$.}
 Let $N \triangleq 2^n$, $n \in \mathbb{N}^*$. Define $U^{1:N} \triangleq V^{1:N} G_n$, $T^{1:N} \triangleq Y^{1:N} G_n$, where $G_n$ is defined in Section \ref{sec:notations}, and define for $\beta < 1/2$ and  $\delta_N \triangleq 2^{-N^{\beta}} $ the sets
\begin{align*}
\mathcal{H}_V &\triangleq \left\{ i \in \llbracket 1, N \rrbracket : H( U^i | U^{1:i-1}) > \delta_N \right\},\\
\mathcal{V}_{V|X} &\triangleq \left\{ i \in \llbracket 1, N \rrbracket : H( U^i | U^{1:i-1} X^{1:N}) > \log |\mathcal{V}| - \delta_N \right\},\\
\mathcal{V}_{Y|V} &\triangleq \left\{ i \in \llbracket 1, N \rrbracket : H( T^i | T^{1:i-1} V^{1:N}) > \log |\mathcal{Y}| - \delta_N \right\},\\
\mathcal{V}_{V|XY} &\triangleq \left\{ i \in \llbracket 1, N \rrbracket : H( U^i | U^{1:i-1} X^{1:N} Y^{1:N}) \right. \\
& \left. \phantom{---------------}> \log |\mathcal{V}| - \delta_N \right\}.
\end{align*}
Note that the sets $\mathcal{H}_V$, $\mathcal{V}_{V|X}$, $\mathcal{V}_{V|XY}$, and $\mathcal{V}_{Y|V}$ are defined with respect to $q_{XYV}$. Similar to the previous sections, we have 
\begin{align*}
\lim_{N\to \infty} |\mathcal{H}_{V}| / N &= H(V), \\ \lim_{N\to \infty} |\mathcal{V}_{V|X}| / N &= H(V|X), \\ \lim_{N\to \infty} |\mathcal{V}_{V|XY}| / N &= H(V|XY), \\ \lim_{N\to \infty} |\mathcal{V}_{Y|V}| / N &= H(Y|V).
\end{align*} 
Note also that $ \mathcal{V}_{V|XY} \subset \mathcal{V}_{V|X} \subset \mathcal{V}_V$. We define  ${\mathcal{F}_1} \triangleq \mathcal{H}_V^c$, ${\mathcal{F}}_2 \triangleq \mathcal{V}_{V|XY}$, $\mathcal{F}_3 \triangleq \mathcal{V}_{V|X}  \backslash \mathcal{V}_{V|XY}$, and $\mathcal{F}_4 \triangleq \mathcal{H}_{V}  \backslash \mathcal{V}_{V|X}$ such that $(\mathcal{F}_1,\mathcal{F}_2,\mathcal{F}_3,\mathcal{F}_4)$ forms a partition of $\llbracket 1 , N \rrbracket$. 
Encoding is performed over $k \in \mathbb{N}^*$ blocks of length $N$. We use the subscript $i \in \llbracket 1,k \rrbracket$ to denote random variables associated to encoding Block $i$. The encoding and decoding procedures are described in Algorithms~\ref{alg:encoding_4} and \ref{alg:encoding_5}, respectively. The functional dependence graph of the coding scheme is depicted in Figure~\ref{figFGD2}.


\begin{figure*}
\centering
  \includegraphics[width=14cm]{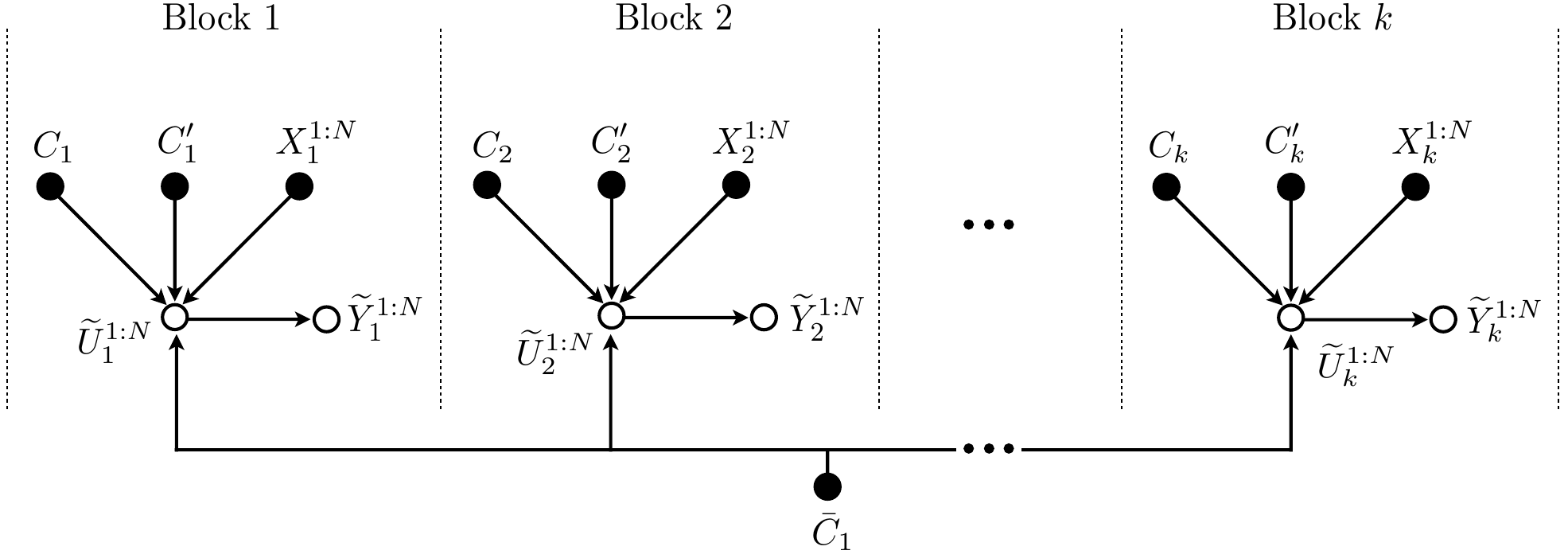}
  \caption{Functional dependence graph of the block encoding scheme for strong coordination. For Block $i$, $(C_i, C_i', \bar{C}_{i}$) is the common randomness shared by Node $1$ and $2$ used at the encoder to form $\widetilde{U}_i$, where $\bar{C}_i = \bar{C}_1$ is reused over all blocks.}
  \label{figFGD2}
\end{figure*}

\begin{algorithm}[]
  \caption{Encoding algorithm at Node $1$ for strong coordination}
  \label{alg:encoding_4}
  \begin{algorithmic} [1]   
    \REQUIRE A vector $C_{1:k}$ of $k|\mathcal{F}_3|$ uniformly distributed symbols over $\llbracket 1 , |\mathcal{V}| \rrbracket$ shared with Node $2$. A vector $\bar{C}_{1}$ of $|\mathcal{F}_2|$ uniformly distributed symbols over $\llbracket 1 , |\mathcal{V}| \rrbracket$ shared with Node $2$ and $X_{1:k}^{1:N}$.
    \FOR{Block $i=1$ to $k$}
    \STATE $\bar{C}_i \leftarrow \bar{C}_1$
    \STATE $\widetilde{U}_i^{1:N}[\mathcal{F}_{2}]\leftarrow \bar{C}_i$
    \STATE $\widetilde{U}_i^{1:N}[\mathcal{F}_{3}]\leftarrow C_i$
    \STATE Given $X_i^{1:N}$, successively draw the remaining components of $\widetilde{U}_i^{1:N}$ according to $\widetilde{p}_{U_i^{1:N}X_i^{1:N}}$ defined by
 \begin{align} \label{true_distrib2}
& \widetilde{p}_{U_i^{j}| U^{1:j-1}_i X_i^{1:N}}(u_i^{j}| \widetilde{U}^{1:j-1}_i X_i^{1:N}) \nonumber \\
& \triangleq 
\begin{cases}
 q_{U^{j}|U^{1:j-1}}(u_i^j|\widetilde{U}_i^{1:j-1}) &\text{ if } j \in  \mathcal{F}_1  , \\
 q_{U^{j}|U^{1:j-1}X^{1:N}}(u_i^j|\widetilde{U}_i^{1:j-1}X_i^{1:N}) &\text{ if } j \in  \mathcal{F}_4   .
\end{cases}
\end{align}    \STATE Transmit $M_i \triangleq \widetilde{U}_i^{1:N}[ \mathcal{F}_4]$ and $C'_i$, the randomness necessary to draw $\widetilde{U}_i^{1:N}[\mathcal{F}_1]$, to Node $2$.
    \ENDFOR
  \end{algorithmic}
\end{algorithm}

\begin{algorithm}[h]
  \caption{Decoding algorithm at Node $2$ for strong coordination}
  \label{alg:encoding_5}
  \begin{algorithmic} [1]   
    \REQUIRE The vectors $C_{1:k}$, $C'_{1:k}$, and $\bar{C}_1$ used in Algorithm~\ref{alg:encoding_4} and $M_{1:k}$.
    \FOR{Block $i=1$ to $k$}
    \STATE $\bar{C}_i \leftarrow \bar{C}_1$
    \STATE $\widetilde{U}_i^{1:N}[\mathcal{F}_{2}]\leftarrow \bar{C}_i$
    \STATE $\widetilde{U}_i^{1:N}[\mathcal{F}_{3}]\leftarrow C_i$  
    \STATE $ \widetilde{U}_i^{1:N}[\mathcal{F}_{4}] \leftarrow M_i$
    \STATE Using ${C}'_i$, successively draw the remaining components of $\widetilde{U}_i^{1:N}$ according to ${q}_{U^j|U^{1:j-1}} $ 
    \STATE $\widetilde{V}_i^{1:N} \leftarrow \widetilde{U}_i^{1:N}G_n$   
    \STATE Channel simulation: Given $\widetilde{V}_i^{1:N}$, successively draw the components of $\widetilde{T}_i^{1:N}$ according to $\widetilde{p}_{T_i^{1:N}V_i^{1:N}}$ defined by
 \begin{align}
& \widetilde{p}_{T_i^{j}| T^{1:j-1}_i V_i^{1:N}}(t_i^{j}| \widetilde{T}^{1:j-1}_i \widetilde{V}_i^{1:N})  \nonumber \\ &  \triangleq 
\begin{cases}
 1/|\mathcal{Y}| &\text{ if } j \in  \mathcal{V}_{Y|V} ,  \\
q_{T^{j}|T^{1:j-1}V^{1:N}}(t_i^j|\widetilde{T}_i^{1:j-1}\widetilde{V}_i^{1:N}) &\text{ if } j \in  \mathcal{V}_{Y|V}^c. 
\end{cases} \label{eqalg5}
\end{align}
    \STATE $\widetilde{Y}_i^{1:N} \leftarrow  \widetilde{T}_i^{1:N} G_n$
    \ENDFOR
  \end{algorithmic}
\end{algorithm}

\subsection{Scheme Analysis} \label{sec:anSC}
Although the coordination metric defined in Section \ref{secdefcoord} is the variational distance, our analysis will be performed with the Kullback-Leibler divergence to highlight similarities with the analysis of channel resolvability in Section \ref{sec:anresolv}. Reverting back to the variational distance is directly obtained with Pinsker's inequality.

The following lemma shows that $\widetilde{p}_{V_i^{1:N}X_i^{1:N}}$, defined by $\widetilde{p}_{X_i^{1:N}} \triangleq q_{X^{1:N}}$ and Equation (\ref{true_distrib2}), approximates $q_{V^{1:N}X^{1:N}}$.
\begin{lem} \label{lemstrongcoord2}
 For any $ i \in \llbracket 1,k\rrbracket$, 
 $$\mathbb{D}(q_{V^{1:N}X^{1:N}}|| \widetilde{p}_{V_i^{1:N} X_i^{1:N}}) \leq \delta_N^{(A)},
$$
 where $\delta_N^{(A)} \triangleq  N \delta_N $.
\end{lem}
\begin{proof}
We have
\begin{align*}
& \mathbb{D}(q_{V^{1:N}X^{1:N}}|| \widetilde{p}_{V_i^{1:N} X_i^{1:N}})  \\
& \stackrel{(a)}{=} \mathbb{D}(q_{U^{1:N}X^{1:N}}|| \widetilde{p}_{U_i^{1:N} X_i^{1:N}})  \\
& \stackrel{(b)}{=} \mathbb{E}_{q_{X^{1:N}}} \left[ \mathbb{D}(q_{U^{1:N}|X^{1:N}}|| \widetilde{p}_{U_i^{1:N}|X_i^{1:N}} )\right]  \\
& \stackrel{(c)}{=} \sum_{j=1}^{N} \mathbb{E}_{q_{U^{1:j-1}X^{1:N}}} \left[ \mathbb{D}(q_{U^{j}|U^{1:j-1}X^{1:N}}|| \widetilde{p}_{U_i^j|U_i^{1:j-1}X_i^{1:N}}) \right] \\
& \stackrel{(d)}{=} \!\!\!\!\! \sum_{j \in \mathcal{F}_{1} \cup \mathcal{F}_{2} \cup \mathcal{F}_{3} } \!\!\!\!\!\!\!   \mathbb{E}_{q_{U^{1:j-1}X^{1:N}}} \left[  \mathbb{D}(q_{U^{j}|U^{1:j-1}X^{1:N}}|| \widetilde{p}_{U^j_i|U_i^{1:j-1}X_i^{1:N}}  ) \right]  \\
& \stackrel{(e)}{=} \sum_{j\in \mathcal{V}_{V|X}} ( \log |\mathcal{V}|  -H(U^{j}|U^{1:j-1}X^{1:N}) ) \\
&  \phantom{mm}+  \sum_{j\in \mathcal{H}_{V}^c}    (H(U^{j}|U^{1:j-1}) - H(U^j|U^{1:j-1}X^{1:N}) )  \\
& \leq |\mathcal{V}_{V|X}| \delta_N +  | \mathcal{H}_{V}^c|   \delta_N   \leq N \delta_N,
\end{align*}
where $(a)$ holds by invertibility of $G_n$, $(b)$ and $(c)$ hold by the chain rule for divergence \cite{Cover91}, $(d)$ holds by \eqref{true_distrib2}, $(e)$ holds by uniformity of $\widetilde{U}_i^{1:N}[\mathcal{F}_2 \cup \mathcal{F}_3]=\widetilde{U}_i^{1:N}[\mathcal{V}_{V|X}]$, and by definition of   $\widetilde{p}_{U^j_i|U_i^{1:j-1}X_i^{1:N}}$ in~(\ref{true_distrib2}).
\end{proof}
\begin{rem}
By Lemma \ref{lemstrongcoord2} and because $|\mathcal{F}_2 \cup \mathcal{F}_3|/N=|\mathcal{V}_{V|X}| / N 
 \xrightarrow{N \to \infty} {H(V|X)}$, note that the encoding algorithm of Section \ref{sec:CSSC} performs a conditional resolvability-achieving random number generation in each block. 
\end{rem}
We now show that strong coordination holds for each block in the following lemma.
\begin{lem} \label{lem_indiv}
For $i \in \llbracket 1 , k \rrbracket $, we have
\begin{align*} 
& \mathbb{D} (\widetilde{p}_{X_{i}^{1:N}Y_{i}^{1:N}}|| q_{X^{1:N}Y^{1:N}} )   \\
& \leq \mathbb{D} (\widetilde{p}_{V_{i}^{1:N}X_{i}^{1:N}Y_{i}^{1:N}}|| q_{V^{1:N}X^{1:N}Y^{1:N}} )\\
&  \leq \delta_N^{(B)},
\end{align*}
 where $\delta_N^{(B)} = O\left(N^{3/2} \delta_N^{1/2} \right)$.
 \end{lem}
\begin{proof}
We have 
\begin{align}
& \mathbb{D} (\widetilde{p}_{V_{i}^{1:N}X_{i}^{1:N}Y_{i}^{1:N}} || q_{V^{1:N}X^{1:N}Y^{1:N}} ) \nonumber   \\ \nonumber
& = \mathbb{D} (\widetilde{p}_{Y_{i}^{1:N}| V_{i}^{1:N}X_{i}^{1:N}} \widetilde{p}_{V_{i}^{1:N}X_{i}^{1:N}} || q_{Y^{1:N}|V^{1:N}X^{1:N}}  q_{V^{1:N}X^{1:N}}) \nonumber  \\ \nonumber
& = \mathbb{D} (\widetilde{p}_{Y_{i}^{1:N}| V_{i}^{1:N}} \widetilde{p}_{V_{i}^{1:N}X_{i}^{1:N}} || q_{Y^{1:N}|V^{1:N}}  q_{V^{1:N}X^{1:N}}) \\ \nonumber
& \stackrel{(a)}{\leq} N\log \left(\frac{1}{\mu_{VXY}} \right) \sqrt{2 \ln 2} \\ \nonumber
& \times \left[ \sqrt{  \mathbb{D} (\widetilde{p}_{Y_i^{1:N}|V_{i}^{1:N}}  q_{V^{1:N}X^{1:N}} || \widetilde{p}_{Y_{i}^{1:N}| V_{i}^{1:N}} \widetilde{p}_{V_{i}^{1:N}X_{i}^{1:N}} ) } \right.\\ \nonumber
&\phantom{-l}+  \sqrt{\mathbb{D} ( \widetilde{p}_{Y_i^{1:N}|V_{i}^{1:N}}  q_{V^{1:N}X^{1:N}}||\widetilde{p}_{Y_i^{1:N}V_{i}^{1:N}} q_{X^{1:N}|V^{1:N}})}  \nonumber \\  \nonumber
&\phantom{-l}+ \left. \sqrt{\mathbb{D} ( q_{Y^{1:N}|V^{1:N}}  q_{V^{1:N}X^{1:N}} ||\widetilde{p}_{Y_i^{1:N}V_{i}^{1:N}} q_{X^{1:N}|V^{1:N}})} \right] \nonumber \\  \nonumber
& = N\log \left(\frac{1}{\mu_{VXY}} \right) \sqrt{2 \ln 2}   \left[ \sqrt{  \mathbb{D} (    q_{V^{1:N}X^{1:N}} ||\widetilde{p}_{V_{i}^{1:N}X_{i}^{1:N}} )} \right. \\
& \phantom{-l} \left. +  \sqrt{\mathbb{D} ( q_{V^{1:N}}|| \widetilde{p}_{V_{i}^{1:N}}  ) } +  \sqrt{\mathbb{D} ( q_{Y^{1:N}V^{1:N}}|| \widetilde{p}_{Y_i^{1:N}V_{i}^{1:N}} )} \right] \nonumber \\ \nonumber
& \stackrel{(b)}{\leq}  N\log \left(\frac{1}{\mu_{VXY}} \right) \sqrt{2 \ln 2} \left[ 2 \sqrt{\delta_N^{(A)} } \right. \\ 
& \phantom{---------} \left. + \sqrt{\mathbb{D} ( q_{Y^{1:N}V^{1:N}}|| \widetilde{p}_{Y_i^{1:N}V_{i}^{1:N}} )} \right], \label{eq:2}
\end{align}
where $(a)$ holds similar to the proof of Lemma \ref{lemnew2} in Appendix~\ref{appdiv}, $(b)$ holds by Lemma \ref{lemstrongcoord2}.  
We bound the right-hand side of~(\ref{eq:2}) by analyzing Step 8 of Algorithm \ref{alg:encoding_5} as follows:
\begin{align}
& \mathbb{D} ( q_{Y^{1:N}V^{1:N}}|| \widetilde{p}_{Y_i^{1:N}V_i^{1:N}} )  \nonumber \\ \nonumber
& \stackrel{(a)}{=} \mathbb{D} ( q_{T^{1:N}V^{1:N}}|| \widetilde{p}_{T_i^{1:N}V_i^{1:N}} )  \nonumber \\ \nonumber
& \stackrel{(b)}{=}  \mathbb{E}_{ q_{V^{1:N}}} \left[ \mathbb{D} ( q_{T^{1:N}|V^{1:N}} || \widetilde{p}_{T_i^{1:N}|V_i^{1:N}} ) \right]  +   \mathbb{D} ( q_{V^{1:N}} || \widetilde{p}_{V_i^{1:N}} )  \\ \nonumber
& \stackrel{(c)}{=}  \mathbb{E}_{ q_{V^{1:N}}} \left[ \mathbb{D} ( q_{T^{1:N}|V^{1:N}} || \widetilde{p}_{T_i^{1:N}|V_i^{1:N}} ) \right]  +   \mathbb{D} ( q_{U^{1:N}} || \widetilde{p}_{U_i^{1:N}} )  \\ \nonumber
& \stackrel{(d)}{\leq}    N \delta_N +  \mathbb{E}_{ q_{V^{1:N}}} \left[ \mathbb{D} ( q_{T^{1:N}|V^{1:N}} ||\widetilde{p}_{T_i^{1:N}|V_i^{1:N}} )  \right] \\ \nonumber
& \stackrel{(e)}{=} \! N \delta_N \!+\! \sum_{j =1}^N   \!\mathbb{E}_{ q_{T^{1:j-1}V^{1:N}}} \!\!\!\left[ \mathbb{D} ( q_{T^{j}|T^{1:j-1}V^{1:N}} ||\widetilde{p}_{T_i^{j}|T_i^{1:j-1}V_i^{1:N}}\!)\! \right]    \\ \nonumber
& \stackrel{(f)}{=} \! \! N \delta_N + \!\!\!\!\!\!\sum_{j \in \mathcal{V}_{Y|V}}  \!\!\!\! \! \mathbb{E}_{ q_{T^{1:j-1}V^{1:N}}} \!\!\!\left[ \mathbb{D} ( q_{T^{j}|T^{1:j-1}V^{1:N}} ||\widetilde{p}_{T_i^{j}|T_i^{1:j-1}V_i^{1:N}} \!)\! \right]    \\ \nonumber
& \stackrel{(g)}{=}  N \delta_N + \sum_{j \in \mathcal{V}_{Y|V}}    (\log |\mathcal{Y}|- H(T^{j}|T^{1:j-1}V^{1:N}))    \\ 
& \stackrel{(h)}{\leq} N \delta_N + |\mathcal{V}_{Y|V}|   \delta_N  \nonumber \\
& \leq 2 N \delta_N, \label{eq:3}
\end{align}
where $(a)$ and $(c)$ hold by invertibility of $G_n$ and the data processing inequality, $(b)$ holds by the chain rule for divergence \cite[Th. 2.5.3]{Cover91}, $(d)$ holds because $\mathbb{D} ( q_{U^{1:N}} || \widetilde{p}_{U_i^{1:N}} ) \leq \mathbb{D} ( q_{U^{1:N}X^{1:N}} ||\widetilde{p}_{U_i^{1:N}X_i^{1:N}} )$ (by the chain rule for divergence and positivity of the divergence) and because $\mathbb{D} ( q_{U^{1:N}X^{1:N}} ||\widetilde{p}_{U_i^{1:N}X_i^{1:N}} ) \leq  N \delta_N$ by the proof of Lemma~\ref{lemstrongcoord2}, $(e)$ holds by the chain rule for divergence, $(f)$ and $(g)$ hold by \eqref{eqalg5}, $(h)$ holds by definition of  $\mathcal{V}_{Y|V}$. Finally, combining (\ref{eq:2}), (\ref{eq:3}) yields the claim.
\end{proof}
Using Lemma \ref{lem_indiv}, we now establish the asymptotic independence of consecutive blocks. 
\begin{lem} \label{lem_adj}
For $i\in \llbracket 2, k \rrbracket$, we have,
$$
\mathbb{D} \left( \widetilde{p}_{X_{i-1:i}^{1:N}Y_{i-1:i}^{1:N}\bar{C}_1} || \widetilde{p}_{Y_{i-1}^{1:N}X_{i-1}^{1:N}} \widetilde{p}_{X_{i}^{1:N}Y_{i}^{1:N} \bar{C}_1} \right)  \leq \delta_N^{(C)},
$$
where $\delta_N^{(C)} = O\left(N^{9/4} \delta_N^{1/4} \right)$.
\end{lem}

\begin{proof}
We reuse the proof of Lemma \ref{lem3} with the substitutions $q_{U^{1:N}} \leftarrow q_{V^{1:N}}$, $q_{Y^{1:N}} \leftarrow q_{X^{1:N}Y^{1:N}}$, $\widetilde{p}_{U_i^{1:N}} \leftarrow \widetilde{p}_{V_i^{1:N}}$, $\widetilde{p}_{Y_i^{1:N}} \leftarrow \widetilde{p}_{X_i^{1:N}Y_i^{1:N}}$. Note that the Markov condition ${X}_{i-1}^{1:N}\widetilde{Y}_{i-1}^{1:N} - \bar{C}_1 -  {X}_{i}^{1:N}\widetilde{Y}_{i}^{1:N}$ holds as seen in Figure \ref{figFGD2}. 
\end{proof}
Using Lemma \ref{lem_adj} we next show the asymptotic independence across all blocks.
\begin{lem} \label{lemAll}
We have
$$
\mathbb{D} \left( \widetilde{p}_{X_{1:k}^{1:N}Y_{1:k}^{1:N}} || \prod_{i =1}^k  \widetilde{p}_{X_{i}^{1:N}Y_{i}^{1:N}} \right) \leq (k-1) \delta_N^{(C)}.
$$
where $\delta_N^{(C)}$ is defined in Lemma \ref{lem_adj}.
\end{lem}
\begin{proof}
We reuse the proof of Lemma \ref{lem5} with the substitutions $\widetilde{p}_{Y_i^{1:N}} \leftarrow \widetilde{p}_{X_i^{1:N}Y_i^{1:N}}$. Note that the Markov condition ${X}_{1:i-2}^{1:N}\widetilde{Y}_{1:i-2}^{1:N} - \bar{C}_1 \widetilde{X}_{i-1}^{1:N} \widetilde{Y}_{i-1}^{1:N} - {X}_{i}^{1:N}\widetilde{Y}_{i}^{1:N}$ holds, as seen in Figure~\ref{figFGD2}. 
\end{proof}
Using Lemmas~\ref{lem_indiv} and \ref{lemAll}, we now show that strong coordination holds over all blocks.
\begin{lem} \label{lemAllC}
We have
$$
\mathbb{D} \left( \widetilde{p}_{X_{1:k}^{1:N}Y_{1:k}^{1:N}} || q_{X^{1:kN}Y^{1:kN}} \right) \leq  \delta_N^{(D)}.
$$
where $\delta_N^{(D)} = O\left(k^{3/2}N^{17/8} \delta_N^{1/8} \right)$.
\end{lem}
\begin{proof}
We reuse the proof of Lemma \ref{lemAllb} with the substitutions $q_{Y^{1:N}} \leftarrow q_{X^{1:N}Y^{1:N}}$, $\widetilde{p}_{Y_i^{1:N}} \leftarrow \widetilde{p}_{X_i^{1:N}Y_i^{1:N}}$. 
\end{proof}
We can now state our final result as follows.
\begin{thm} \label{ThSC}
\textcolor{black}{The coding scheme described in Algorithms~\ref{alg:encoding_4} and~\ref{alg:encoding_5}, which operate over $k$ blocks of length $N$, achieves the two-node network strong coordination capacity region of Theorem \ref{th_c} for an arbitrary target distribution $q_{XY}$ over $\mathcal{X} \times \mathcal{Y}$, where $|\mathcal{Y}|$ is a prime number.  The coding scheme is explicit with complexity in $O(kN \log N)$.}
\end{thm}
\begin{proof}
We prove that the communication rate and the common randomness rate are optimal. 
The common randomness rate is
\begin{align}
  \frac{|\bar{C}_1 | + |C_{1:k}|}{kN}
 & = \frac{ |\mathcal{V}_{V|XY}| + k |\mathcal{V}_{V|X} \backslash \mathcal{V}_{V|XY}| }{kN}  \nonumber \\ \nonumber
 & = \frac{ |\mathcal{V}_{V|XY}| }{kN} + \frac{|\mathcal{V}_{V|X} | - | \mathcal{V}_{V|XY}| }{N} \\ \nonumber
 & \xrightarrow{N \to \infty} I(V;Y|X) + \frac{H(V|XY)}{k} \\
 &  \xrightarrow{k \to \infty} I(V;Y|X) , \label{eqrate1}
 \end{align}
 where we have used \cite[Lemma 7]{Chou14c}. 
 
Next, we determine the communication rate. Observe first that for any $i \in \llbracket 1 , k \rrbracket$,
  $$
  \lim_{N \to \infty} \frac{|C_i'|}{N}  = \lim_{N \to \infty} \frac{1}{N}\sum_{j \in \mathcal{H}_{V}^c} H(\widetilde{V}^{j}_i|\widetilde{V}^{1:j-1}_i) = 0,
  $$
which can be proved using Lemma \ref{lem_indiv} similar to the proof of Theorem \ref{ThChr}. 
Hence, the communication rate is
\begin{align}
  \frac{ |\widetilde{U}_{1:k}^{1:N} [\mathcal{F}_4]|}{kN}
 & = \frac{  k |\mathcal{F}_4| }{kN}  \nonumber  \\ \nonumber 
 & =  \frac{|\mathcal{V}_{V} | - | \mathcal{V}_{V|X}| }{N} \\  
 & \xrightarrow{N \to \infty} I(V;X)  , \label{eqrate2}
 \end{align}
 where the limit holds by  \cite[Lemma 7]{Chou14c}. 

We also have that the communication rate and the common randomness rate sum to
\begin{align*}
I(V;X) + I(V;Y|X) = I(V;XY),
\end{align*}
which together with \eqref{eqrate1} and \eqref{eqrate2} recovers the bounds of Theorem \ref{th_c}. This result along with Lemma \ref{lemAllC} allow us to conclude.
\end{proof}

\begin{rem}
The rate of local randomness $H(Y|V)$ used at Node 2 is optimal. It is possible to extend the converse proof of \cite{Cuff10} and include the local randomness rate $R_L$ at Node 2 to show that the strong coordination capacity region is 
\begin{align*}
& \mathcal{R}_{\textup{SC}}(q_{XY}) \\ 
& \triangleq \!\!\!\!\!\!\!\!  \smash{ \bigcup_{ \substack{X \to \ V \to Y \\ |\mathcal{V}|\leq |\mathcal{X}||\mathcal{Y}| +1}} \!\!\!  \!\! \!\! \{ (R,R_0,R_L) :  R + R_L \geq I(X;V) + H(Y|V), }\\
& \phantom{--------} \smash{ R +R_0 +R_L \geq I(XY;V) + H(Y|V) \}.}
\end{align*}
\end{rem}
\section{Concluding remarks} \label{secc}
\textcolor{black}{
We have demonstrated the ability of polar codes to provide solutions to problems related to soft covering. Specifically, we have proposed an explicit and low-complexity coding scheme for channel resolvability by relying on (i) conditional resolvability-achieving random number generation and (ii) randomness recycling through block-Markov encoding. As discussed in the introduction, our coding scheme generalizes previous  explicit coding schemes that achieve channel resolvability but were restricted to uniform distributions or symmetric channels.
}

\textcolor{black}{
Furthermore, by leveraging the coding scheme for channel resolvability, we have proposed explicit and low-complexity polar coding schemes that achieve the capacity regions of empirical coordination and strong coordination in two-node networks. 
}

\textcolor{black}{
Note that all our coding schemes require that the cardinality of the alphabet of actions at Node 2 (for Sections \ref{sec:EmpCord}, \ref{sec:StrongCord}) or of the channel input alphabet (for Section \ref{sec:resolvability}) be a prime number. This assumption could be removed if \cite[Lemma 6,7]{Chou14c} used throughout our proofs to establish various rates could be extended to  alphabet with arbitrary cardinalities. Such an extension of \cite[Lemma 6]{Chou14c} is provided in~\cite{Sasoglu11}, however, the problem of obtaining a similar extension for \cite[Lemma~7]{Chou14c} remains open.} 


\appendices
\section{Simple relations for the Kullback-Leibler divergence} \label{appdiv}
	
\textcolor{black}{We provide in this appendix simple relations satisfied by the Kullback-Leibler divergence. More specifically, Lemma \ref{lemnew1} allows us to characterize the symmetry of the Kullback-Leibler divergence around zero. Lemma~\ref{lemnew} allows us to express independence in terms of mutual informations. Lemma \ref{lemnew2} describes a relation similar to the triangle inequality for small values of the Kullback-Leibler divergence.  Finally, Lemma \ref{corent} provides an upper-bound on the difference of the entropy of two random variables defined over the same alphabet. In the proofs of the following lemmas, we only need to consider the case where all the Kullback-Leibler divergences are finite, since the lemmas trivially hold when any of the Kullback-Leibler divergences appearing in the lemmas is infinite (see Section \ref{sec:notations} for the convention used in the definition of the Kullback-Leibler divergence).}

\begin{lem} \label{lemnew1}
Let $p$, $q$ be two distributions over the finite alphabet $\mathcal{X}$. We have
$$
\mathbb{D}\left(p \Vert q \right) \leq  \log \left(\frac {1 }{\mu_{q}} \right) \sqrt{2 \ln 2} \sqrt{\mathbb{D}\left(q || p \right)},
$$	
where $\mu_q \triangleq \displaystyle\min_{x \in \mathcal{X}} q(x)$.
\end{lem}
\begin{proof}
The result follows from Pinsker's inequality and the following inequality~\cite[Eq. (323)]{sason2015bounds}.
\begin{align}
		\mathbb{D}\left(p \Vert q \right)  \label{equationim}
		 \leq \log \left(\frac {1 }{\mu_{q}} \right)\mathbb{V}\left(p , q \right).
	\end{align}
\end{proof}
\begin{lem} \label{lemnew}
Let $\left(X_i \right)_{i\in \llbracket 1 ,k \rrbracket}$ be arbitrary discrete random variables with joint probability $p_{X_{1:k}}$. We have
$$
\mathbb{D}\left(p_{X_{1:k}} \Vert \prod_{i=1}^k p_{X_{i}} \right) = \sum_{i=2}^{k} I(X_i ; X_{1:i-1}).
$$	
\end{lem}
\begin{proof}
We have by the chain rule for relative Kullback-Leibler divergence \cite{Cover91}
	\begin{align*}
		\mathbb{D}\left(p_{X_{1:k}} \Vert \prod_{i=1}^k p_{X_{i}} \right) 
   	    & = \sum_{i=1}^k \mathbb{E}_{X_{1:i-1}} \left[ \mathbb{D}\left(p_{X_{i}|X_{1:{i-1}}} \Vert p_{X_{i}} \right) \right]\\
   	       	    & = \sum_{i=1}^k   \mathbb{D}\left(p_{X_{1:i}} \Vert p_{X_{1:i-1}} p_{X_{i}} \right)\\
   	    & = \sum_{i=2}^{k} I(X_i ; X_{1:i-1}). \tag*{\qedhere}
     	\end{align*}
\end{proof}

\begin{lem} \label{lemnew2}
Let $p$, $q$, and $r$ be distributions over the finite alphabet $\mathcal{X}$. We have
\begin{align*}
	\mathbb{D}\left(p \Vert q \right)  & \leq \log \left(\frac{1}{\mu_q} \right) \sqrt{2 \ln 2} \left[ \sqrt{ \min(\mathbb{D}\left(p \Vert r\right), \mathbb{D}\left(r \Vert p \right))} \right. \\
	& \phantom{---------} + \left. \sqrt{ \min(\mathbb{D}\left(q \Vert r\right), \mathbb{D}\left(r \Vert q \right))} \right],
	\end{align*}
	where $\mu_q \triangleq \displaystyle\min_{x \in \mathcal{X}} q(x)$.
\end{lem}
\begin{proof}
The result follows from \eqref{equationim}, the triangle inequality for the variational distance, and Pinsker's inequality.
\end{proof}
\begin{lem} \label{corent}
Let $p$, $q$ be distributions over the finite alphabet $\mathcal{X}$. Let $H(p)$ and $H(q)$ denote the Shannon entropy associated with $p$ and $q$, respectively. Provided that $\mathbb{D}(p||q)$ or $\mathbb{D}(q||p)$ is small enough, we have 
\begin{align*}
& |H(q) - H(p)| \leq  D \log \left( \frac{|\mathcal{X}|}{D}\right),
\end{align*}
where $D \triangleq \sqrt{2 \ln2} \sqrt{ \min ( \mathbb{D}(p||q), \mathbb{D}(q||p)) }$.
\end{lem}
\begin{proof}
The result follows by \cite[Lemma 2.7]{bookCsizar}, Pinsker's inequality, and because $x \mapsto -x \log x$ is increasing over $]0,e^{-1}]$.
\end{proof}
\section{Proof of Remark \ref{remB}} \label{AppA}
We first observe that $|\mathcal{V}_{X}| / N 
 \xrightarrow{N \to \infty} {H(X)}$. Then, for $i \in \llbracket 1 , k \rrbracket$, \begin{align*}
&\mathbb{D} (\widetilde{p}_{Y_i^{1:N}} || q_{Y^{1:N}} ) \\
& \stackrel{(a)}{\leq} N \log \left( \frac {1 }{\mu_{q_Y}} \right) \sqrt{2 \ln2} \sqrt{\mathbb{D} \left( q_{Y^{1:N}}|| \widetilde{p}_{Y_i^{1:N}}    \right)}\\
& \stackrel{(b)}{\leq} N\log \left(\frac {1 }{\mu_{q_Y}} \right)\sqrt{2 \ln2} \sqrt{\delta_N^{(1)}},
\end{align*}
where $(a)$ holds by Lemma \ref{lemnew1} in Appendix \ref{appdiv} with $\mu_{q_{Y^{1:N}}} =  \mu_{q_Y}^N$, $(b)$ holds by the chain rule  and Lemma~\ref{lem0}.

\bibliographystyle{IEEEtran}
\bibliography{polarwiretap}
\end{document}